\documentclass{article}

\usepackage{graphicx} 
\usepackage{amssymb,comment,algorithm,algorithmic,bbm,stmaryrd,gastex}
\usepackage{amsmath, amsthm}
\usepackage{multirow,rotating}
\usepackage{amsfonts}
\usepackage{amssymb}
\usepackage{times}
\usepackage{pifont}
\usepackage{wrapfig}
\usepackage{cite}
\usepackage[UKenglish]{babel}

\newtheorem{claimN}{Claim}
\newcommand{\cut}[1]{}

\newenvironment{reminder}[1]{\medskip

\noindent {\bf Reminder of Theorem #1.  }\em}{}

\def \QED {\hfill{$\Box$}}
\DeclareMathAlphabet {\mathpzc}{OT1}{pzc}{m}{it}

\newenvironment{proofof}[1]{\noindent {\em Proof of #1.  }}{\QED}

\newif\ifFull 

\newenvironment{reminderclaim}[1]{\medskip
\noindent {\bf Reminder of Claim #1.  }\em}{}
\newtheorem{definition}{Definition}
\newtheorem{theorem}{Theorem}
\newtheorem{lemma}{Lemma}
\newtheorem{remark}{Remark}

\newcommand{\acknowledge}[1]{} 

\newcommand{\BndMemRegMin}{\mathsf{BW}}
\newcommand{\RegMinImperfInfo}{\mathsf{ EXBWII}}
\newcommand{\RegMinPerfInfo}{\mathsf{ EXBW}}
\newcommand{\SampleAlgorithm}{\mathbf{Sample}\left(\PerfectInfoExpertSet\right)}
\newcommand{\FixedDefender}{F}
\newcommand{\KAdaptiveDefenderSet}{ \mathcal{A}^K_D}
\newcommand{\KAdaptiveAdversarySet}[1]{\mathcal{A}^{#1}_A}
\newcommand{\AverageRegret}[1]{\bar{R}_{#1}}
\newcommand{\AveragePay}{\bar{P}}
\newcommand{\Pay}{P}
\newcommand{\KAdaptiveAdversary}[1]{A_{#1}}
\newcommand{\Adversary}{A}
\newcommand{\Defender}{D}
\newcommand{\GameClass}{\mathcal{G}}

\newcommand{\BoundedMemoryGamesImperfectInfo}{\mathcal{G}}
\newcommand{\DefenderActions}{\mathcal{X}_D}
\newcommand{\AdversaryActions}{\mathcal{X}_A}
\newcommand{\OutcomeSet}{\mathcal{O}}
\newcommand{\Consistent}[1]{\mathcal{C}\left(#1\right)}
\newcommand{\ConsistentAll}{\mathcal{C}}
\newcommand{\ReachableTrace}[1]{\mathcal{R}\left( #1 \right)}
\newcommand{\RepeatedGame}[2]{\rho\left(#1,#2 \right)}
\newcommand{\PerfectInfoExpertSet}{\mathcal{E}}

\DeclareGraphicsExtensions{.jpg, .pdf, .eps}

\title{Adaptive Regret Minimization in Bounded-Memory Games \thanks{This work was partially supported by the U.S. Army Research Office contract ``Perpetually Available and Secure Information Systems" (DAAD19-02-1-0389) to Carnegie Mellon CyLab, the NSF Science and Technology Center TRUST, the NSF CyberTrust grant ``Privacy, Compliance and Information Risk in Complex Organizational Processes,'' the AFOSR MURI ``Collaborative Policies and Assured Information Sharing,'' and HHS Grant no. HHS 90TR0003/01. Jeremiah Blocki was also partially supported by a NSF Graduate Fellowship. Arunesh Sinha was also partially supported by the CMU CIT Bertucci Fellowship. The views and conclusions contained in this document are those of the authors and should not be interpreted as representing the official policies, either expressed or implied, of any sponsoring institution, the U.S. government or any other entity.}}
\author{
Jeremiah Blocki \\ Carnegie Mellon University\\  \texttt{jblocki@cs.cmu.edu} \and 
Nicolas Christin \\ Carnegie Mellon University\\   \texttt{nicolasc@cmu.edu} \and 
Anupam Datta \\ Carnegie Mellon University\\   \texttt{danupam@cmu.edu} \and
Arunesh Sinha \\Carnegie Mellon University\\  \texttt{aruneshs@cmu.edu}
 }

\begin{document}

\maketitle

%
%
\begin{abstract}
Organizations that collect and use large volumes of personal information often use security audits to protect data subjects from inappropriate uses of this information by authorized insiders. In face of unknown incentives of employees, a reasonable audit strategy for the organization (defender) is one that minimizes his regret. While regret minimization has been extensively studied in repeated games, a repeated game cannot capture the full complexity of the interaction between the organization (defender) and an insider (adversary) that arises from dependence of rewards and actions on history. We introduce a richer class of games called \emph{bounded memory games}, which can provide a more accurate model of the audit process. The standard notion of regret for repeated games is no longer suitable because actions and rewards can depend on the history of play. To account for this generality, we introduce the notion of {\em
$k$-adaptive regret}, which compares the reward obtained by playing actions prescribed by the algorithm against a
hypothetical {\em $k$-adaptive adversary} with the reward obtained by the best expert in hindsight against the same adversary. Roughly, a hypothetical $k$-adaptive adversary adapts her strategy to the defender's actions exactly as the real adversary would within each window of $k$ rounds. A $k$-adaptive adversary is a natural model for temporary employees who stay for a certain
number of audit cycles and are then replaced by a different person. Our definition is parametrized by a set of experts, which can include both fixed and adaptive defender strategies.

$\quad$ We investigate the inherent complexity of and design algorithms for adaptive regret minimization in bounded memory games of perfect and imperfect information. We prove a hardness result showing that, with imperfect information, any $k$-adaptive regret minimizing algorithm (with fixed strategies as experts) must be inefficient unless {\sf NP}
$=$ {\sf RP} even when playing against an oblivious adversary. In contrast, for bounded memory games of perfect and imperfect information we present approximate $0$-adaptive regret minimization algorithms against an oblivious adversary running in time $n^{O\left(1\right)}$.
\end{abstract}

\section{Introduction} 
\label{sec:Introduction}
%
%
%
%

Online learning algorithms that minimize regret provide strong guarantees in
situations that involve repeatedly making decisions in an uncertain
environment. There is a well developed theory for regret minimization in repeated games\cite{blum2007learning}. The goal of this paper is to study regret minimization for a richer class of settings. As a motivating example consider a hospital (defender) where a \emph{series} of temporary employees or business affiliates (adversary) access patient records for legitimate purposes (e.g., treatment or payment) or inappropriately (e.g., out of curiosity about a family member or for financial gain). The hospital conducts audits to catch the violators, which involves expending resources in the form of time spent in human investigation.
On the other hand, violations that are missed internally and caught externally (by Government audits, patient complaints, etc.) also result in various losses such as repuation loss, loss due to litigation, etc. The hospital wants to minimize its overall loss by balancing the cost of audits with the risk of externally detected violations. In these settings with unknown adversary incentives, a reasonable strategy for the defender is one that minimizes her regret.

\cut{As a motivating example consider a store that sells perishable goods (e.g., bread, milk, eggs) and faces a series of different customers every $k$ rounds. The store manager may be uncertain about the future demand for certain goods. Nevertheless, each round the store manager must decide how to price the goods and whether or not to restock. Another motivating example involves developing effective auditing strategies in an adversarial environment: Consider a hospital (defender) where a series of different employees or business affiliates (adversary) access patient records for legitimate purposes (e.g., treatment or payment) or inappropriately (e.g., out of curiosity about a family member or for financial gain). The hospital wants to minimize its overall loss by balancing the cost of audits with the risk of externally detected violations. 
In these settings, a reasonable strategy for the defender (manager) is one that
minimizes her regret. }\cut{An auctioneer who repeatedly sells goods to different groups of bidders will want to learn from past experience, even if the bidders are different in every auction \cite{chakraborty2008online}. }

Modeling this interaction as a repeated game of imperfect information is challenging because this game has two additional characteristics that are not captured by a repeated game model: (1) {\em History-dependent rewards}: The payoff function depends not only on the current outcome but also on previous outcomes. For example, when a violation occurs the hospital might experience a greater loss if other violations have occured in recent history. (2) {\em History-dependent actions}: Both players may \emph{adapt} their strategies based on history. For example, if many violations have been detected and punished in recent history then a rational employee might choose to lay low rather than committing another violation.

Instead, we capture this form of history dependence by introducing \emph{ bounded memory games}, a subclass of stochastic games \footnote{Stochastic games~\cite{shapley1953stochastic}\acknowledge{We thank Avrim Blum for suggesting stochastic games as a suitable model.} are expressive enough to model history dependence. However, there is no regret minimization algorithm for the \emph{general class of stochastic games}. While we do not view this result as surprising or novel, we include it in Appendix \ref{apx:StochasticGames} for completeness.}. In each round of a two-player bounded-memory-$m$ game, both players simultaneously play an action, observe an outcome and receive a reward. In contrast to a repeated game, the payoffs may depend on the state of the game. In contrast to a {\em general} stochastic game, the rewards may \emph{only} depend on the outcomes from the last $m$ rounds (e.g., violations that were caught in the last $m$ rounds) as well as the actions of the players in the current round. 

In a bounded memory game, the standard notion of regret for a repeated game is not suitable because the adversary may adapt her actions based on the history of play.  To account for this generality, we introduce (in Section~\ref{sec:regretMin}) the notion of {\em
$k$-adaptive regret}, which compares the reward obtained by playing actions prescribed by the algorithm against a
hypothetical {\em $k$-adaptive adversary} with the reward obtained by the best expert in hindsight against the same adversary. Roughly, a hypothetical $k$-adaptive adversary plays exactly the same actions as the real adversary except in the last $k$ rounds where she adapts her strategy to the defender's actions exactly as the real adversary would.   When $k = 0$, this definition coincides with the
standard definition of an \emph{oblivious adversary} considered in defining
regret for repeated games.  When $k = \infty$ we get a \emph{fully adaptive
adversary}. 
A $k$-adaptive adversary is a natural model for temporary employees (e.g., residents, contractors) who stay for a certain
number of audit cycles and are then replaced by a different person.
Our definition is parameterized by a set of experts, which can include both
fixed and adaptive defender strategies. In section \ref{sec:Audit} we use the example of a police chief enforcing the speed limit at a popular tourist destination (or a hospital auditing accesses to the patient records made by residents) to illustrate the power of $k$-adaptive regret minimization when the defender plays against a series of temporary adversaries. 

\cut{We show how to model the audit interaction using bounded memory games, and we illustrate the notion of $k$-adaptive regret with several examples. We believe that bounded memory $k$-adaptive regret may also be natural choice in many other settings (e.g., the series of different customers in the chainstore game, different bidders in a repeated auction).}

Next, we investigate the inherent complexity of and design algorithms for adaptive 
regret minimization in bounded-memory games of perfect 
and imperfect information. 
Our results are summarized in Table \ref{table:BoundedMemory}.
We prove a hardness result (Section~\ref{sec:hardness}; Theorem
\ref{thm:hardness}) showing that, with imperfect information, any $k$-adaptive
regret minimizing algorithm (with fixed strategies as experts) must be inefficient unless {\sf NP} $=$ {\sf RP}
even when playing against an oblivious adversary and even when $k = 0$.  In
fact, the result is even stronger and applies to any $\gamma$-approximate
$k$-adaptive regret minimizing algorithm (ensuring that the regret bound
converges to $\gamma$ rather than $0$ as the number of rounds $T \rightarrow
\infty$) for $\gamma < \frac{1}{8n^{\beta}}$ where $n$ is the number of states in the game and $\beta
>0$. Our hardness reduction from MAX3SAT uses the state 
of the bounded-memory game and the history-dependence of rewards in a critical way. \cut{ Using a slightly stronger complexity-theoretic assumption, we improve
this bound to include any value of $\gamma$  less than $\frac{1}{8 \log^2 n}$.
Technically, the hardness results are established
via reduction from MAX3SAT. In the reduction, each of the 
exponentially many possible variable assignments corresponds to an expert (fixed strategy); 
performing as well as the best fixed strategy in the game corresponds to satisfying as many 
clauses as the best variable assignment. }

We present an inefficient $k$-adaptive regret minimizing algorithm by reducing
the bounded-memory game to a repeated game. The algorithm is inefficient
for bounded-memory games when the number of experts is exponential in the number of states of the game (e.g., if all fixed strategies are experts). In contrast, for bounded-memory games of perfect information, we present an efficient
$n^{O\left(1/\gamma\right)}$ time $\gamma$-approximate $0$-adaptive regret
minimization algorithm against an oblivious adversary for any constant $\gamma
> 0$ (Section~\ref{sec:algorithm};Theorem \ref{thm:PerfectInformation}). We also show how this algorithm can be adapted
to get an efficient  $\gamma$-approximate $0$-adaptive regret minimization
algorithm for bounded-memory games of imperfect information
(Section~\ref{sec:algorithm};Theorem \ref{thm:ImperfectInformationApproximateRegret}). 
The main novelty in these algorithms is an implicit weight representation for an 
exponentially large set of adaptive experts, which includes all fixed strategies.

\begin{table}[h] \begin{tabular}{|p{4.5cm}||p{3.5cm}|p {3.5cm}|} \hline & Imperfect
Information & Perfect Information \\ \hline \hline Oblivious Regret $(k=0)$ &
Hard (Theorem \ref{thm:hardness}) & APX (Theorem \ref{thm:PerfectInformation})
\\ & APX (Theorem \ref{thm:ImperfectInformationApproximateRegret})  &     \\ \hline 

$k$-Adaptive Regret $( k \geq 1)$&   Hard (Theorem \ref{thm:hardness}) &  Hard (Remark \ref{remark:HardnessExtension} )\\ \hline
Fully Adaptive Regret $(k=\infty)$ & X  (Theorem \ref{thm:stochasticImpossible}) & X (Theorem \ref{thm:stochasticImpossible}) \\ \hline \end{tabular}

\caption{ {\bf Regret Minimization in Bounded Memory Games}\newline  X - no
regret minimization algorithm exists  \newline  Hard - unless {\sf NP} $=$ {\sf
RP} no regret minimization algorithm is efficiently computable \newline APX -
efficient approximate regret minimization algorithms exist.  \newline }
\label{table:BoundedMemory}

\end{table}

\section{Related Work} \label{sec:related}
A closely related work is the Regret Minimizing Audit (RMA) mechanism of Blocki et al. \cite{blocki2011regret}, which uses a repeated game model for the audit process. RMA deals with history-dependent rewards under certain 
assumptions about the defender's payoff function, but it does not consider history-dependent actions. While RMA provides strong performance guarantees for the defender against a byzantine adversary, the performance of RMA may be far from optimal when the adversary is rational (or nearly rational). In subsequent work the same authors \cite{BCDS12} introduced a model of a nearly rational adversary who behaves in a rational manner most of the time. A nearly rational adversary can usually be deterred from committing policy violations by high inspection and punishment levels. They suggested that the defender commit to his strategy before each audit round (e.g., by publically releasing its inspection and punishment levels) as in a Stackelberg game \cite{von2011market}. However, the paper gives no efficient algorithm for computing the Stackelberg equilibrium. 

More recent work by Blocki et al. introduced the notion of Audit Games \cite{AuditGames} --- a simplified game theoretic model of the audit process in which the adversary is purely rational (unlike the nearly rational adversary of \cite{BCDS12}). Audit Games generalize the model of Security Games \cite{tambe} by including punishment level as part of the defenders action space. Because the punishment parameter introduces quadratic constraints into the optimization problem that must be solved to compute the Stackelberg equilibria, this apparently small change makes it difficult to find the Stackelberg equilibria. The primary technical contribution of \cite{BCDS12} is an efficient algorithm for computing the Stackelberg equilibrium of Audit Games. There are two potential advantages of the $k$-adaptive regret framework compared with the Stackelberg equilibria appraoch: (1) The $k$-adaptive regret minimization algorithm can be used even if the adversary's incentives are unknown, and (2) A $k$-adaptive adversary is a better model for a short term adversary (e.g., contractors, tourists) who may not informed about the defender's policy --- and therefore may not even know what the `rational' best response is in a Stackelberg game. See section \ref{sec:Audit} for additional discussion.

Stochastic games were defined by Shapley                  
\cite{shapley1953stochastic}. Much of the work on stochastic games has focused on finding and computing equilibria for these games         
\cite{shapley1953stochastic,mertens1981stochastic}. There has been lot of work in regret minimization for repeated games\cite{blum2007learning}. Regret minimization in stochastic games has not been the subject of much research. Papadimitriou and Yannakakis showed that many natural optimization problems relating to stochastic games are hard \cite{papadimitriou1999complexity}. These results don't apply to bounded memory games. Golovin and Krause recently showed that a simple greedy algorithm can be used when a stochastic optimization problem satisfies a property called adaptive submodularity \cite{golovin2010AdaptiveSubmodularity}. In general, bounded memory games do not satisfy this property. Even-Dar, et al., show 
that regret minimization is possible for a class of stochastic 
games (Markov Decision Processes) in which the adversary chooses the reward function at each state but does not 
influence the transitions\cite{even2005experts}. They also prove that if the adversary controls the reward function and the transitions, then it is {\sf NP-Hard} to even approximate the best fixed strategy. Mannor and Shimkin \cite{mannor2003empirical} show that if the adversary completely controls the transition model (a Controlled Markov Process) then it is possible to separate the stochastic game into a series of matrix games and efficiently minimize regret in each matrix game.  Bounded-memory games are a different subset of stochastic games where the transitions and rewards are influenced by both players.  While our hardness proof shares techniques with  Even-Dar, et al.,\cite{even2005experts}, there are significant differences that arise from the bounded-memory nature of the game. We provide a detailed comparison in Section \ref{sec:hardness}.

In a recent paper, Even-Dar, et al., \cite{evenlearning} handle a few specific global cost functions related to load balancing. These cost functions depend on history. In their setting, the adversary obliviously plays actions from a joint distribution. In contrast, we consider arbitrary cost functions with bounded dependence on history and adaptive adversaries. 

Takimoto and Warmuth \cite{takimoto2003path} developed an efficient online shortest path algorithm. In their setting the experts consists of all fixed paths from the source to the destination. Because there may be exponentially many paths their algorithm must use an implicit weight representation. Awerbuch and Kleinberg later provided a general framework for online linear optimization \cite{awerbuch2008online}. In our settings, an additional challenge arises because \emph{experts adapt to adversary actions}. See Section~\ref{sec:algorithm} for a more detailed comparison.

Farias, et al., \cite{Farias2006} introduce a special class of adversaries that they call ``flexible" adversaries. A defender playing against a flexible adversary can minimize regret by learning the average expected reward of every expert. Our work differs from theirs in two ways. First, we work with  a stochastic game as opposed to a repeated game. Second, our algorithms can handle a sequence of different $k$-adaptive adversaries instead of learning a single flexible adversary strategy. A single $k$-adaptive strategy is flexible, but a sequence of $k$-adaptive adversaries is not. 


%
%
\section{Preliminaries} 
\label{sec:preliminaries}

\label{subsec:BoundedMemoryGames}
Bounded-memory games are a sub-class of stochastic games, in which outcomes and states satisfy certain properties. Formally, a two-player stochastic  
game between an attacker~$A$ and a defender~$D$ is given by $(\DefenderActions, \AdversaryActions, \Sigma, \Pay, \tau)$, where $\AdversaryActions$ and $\DefenderActions$ are the actions spaces for players~$A$ and $D$,           
respectively, $\Sigma$ is the state space, $\Pay:  
\Sigma \times \DefenderActions \times  \AdversaryActions \rightarrow [0,1]$ is the payoff function and $\tau: \Sigma \times \DefenderActions \times  \AdversaryActions \times \{0,1\}^* \rightarrow \Sigma$ is the randomized transition  function linking the different states. Thus, the payoff during round $t$ depends on the current state          
(denoted $\sigma^t$) in addition to the actions of the defender ($d^t$) and the adversary ($a^t$). We use $n=\left|\Sigma \right|$ to denote the number of states.

A {\em bounded-memory game with memory $m$} ($m \in \mathbb{N}$) is a stochastic game with the following properties:
(1) The game satisfies independent outcomes, and
(2) The states $\Sigma = \mathcal{O}^m$ encode the last $m$ outcomes, i.e., 
$ \sigma^i = \left( O^{i-1},\ldots, O^{i-m} \right)$. An outcome of a given round of play is a signal observed by both players (called ``public signal'' in games~\cite{fundenberggame}).  Outcomes depend probabilistically on the actions taken by the players. We use $\OutcomeSet$ to denote the outcome space and $O^t \in \OutcomeSet$ to denote the outcome during round $t$. We say that a game satisfies {\em independent outcomes} if $O^t$ is conditionally independent of $\left(O^1,...,O^{t-1}\right)$ given $d^t$ and $a^t$. Notice that the defender and the adversary in a game with independent outcomes may still select their actions based on history. However, once those actions have been selected, the outcome is independent of the game history. Note that a repeated game is a bounded-memory-$0$ game (a bounded-memory game with memory $m = 0$).

A game in which players only observe the outcome $O^t$ after round $t$ but not the actions taken during a round is called an  {\em imperfect information} game. If both players also observe the actions then the game is a {\em perfect information} game.

The {\em history} of a game $H = \left(O^1,O^2,\ldots,O^i,\ldots,O^t\right) \ ,$  is the sequence of outcomes. We use $H_k$ to denote the $k$ most recent outcomes in the game (i.e., $H_k = \left(O^{t-k+1}; \ldots; O^{t}\right)$), and 
$t=|H|$ to denote the total number of rounds played. We use $H^i$ to
denote the first $i$ outcomes in a history (i.e., $H^i = \left(O^1, \ldots, O^i\right)$), and $H;H\sp{\prime}$ to denote concatenation of histories $H$ and $H\sp{\prime}$.

A {\em fixed strategy} for the defender in a stochastic game is a function $f:\Sigma \rightarrow \DefenderActions$ mapping each state to a fixed action. $\FixedDefender$ denotes the set of all fixed strategies.

\section{Definition of Regret}  \label{sec:regretMin}
As discussed earlier, regret minimization in repeated games has         
received a lot of attention \cite{blum2005external}. Unfortunately,    
the standard definition of regret in repeated games does not directly   
apply to stochastic games. In a repeated game, regret is computed by    
comparing the performance of the defender strategy $\Defender$ with   
the performance of a fixed strategy~$f$. However, in a stochastic       
game, the actions of the defender and the adversary in round~$i$ 
influence payoffs in each round for the rest of the game. Thus, it      
is unclear how to choose a meaningful fixed strategy~$f$ as a      
reference. We solve this conundrum by introducing an adversary-based definition of regret.                                       

\subsection{Adversary Model} \label{subsec:AdversaryModel}
We define a parameterized class of adversaries called $k$-adaptive adversaries, where the parameter $k$ denotes the level of adaptiveness of the adversary. 
Formally, we say that an agent is {\em $k$-adaptive} if its strategy $\Adversary(H)$ is defined by a function $f:\OutcomeSet^* \times \mathbb{N}\rightarrow \AdversaryActions$ such that $\Adversary(H) = f\left(H_i,t\right)$, where $i = t \mod \left(k + 1 \right)$. Recall that $H_i$ is the $i$ most recent outcomes, and $t = |H|$.

As special cases we define an {\em oblivious 
adversary} $\left(k = 0\right)$ and a {\em fully adaptive adversary} $\left(k = \infty\right)$. Oblivious adversaries essentially play        
without any memory of the previous outcomes. Fully adaptive             
adversaries, on the other hand, choose their actions based on the       
entire outcome history since the start of the game. $k$-adaptive        
adversaries lie somewhere in between. At the start of the game, they    
act as fully adaptive adversaries, playing with the entire outcome      
history in mind. But, different from fully adaptive adversaries, every  
$k$~rounds, they ``forget'' about the entire history of the game and    
act as if the whole game was starting afresh. As discussed earlier,
there are numerous practical instances where $k$-adaptive adversaries
are an appropriate model; for instance, in games in which one player
(e.g., a firm) has a much longer length of play than the adversary
(e.g., a temporary employee), it may be judicious to model the adversary
as $k$-adaptive.
In particular, 
$k$-adaptive adversaries are similar to the notion of ``patient''       
players in long-run games discussed by \cite{LongRun}. 
Their notion of ``fully patient'' players correspond to fully         
adaptive adversaries, ``myopic'' players correspond to oblivious        
adversaries, and ``not myopic but less patient'' players correspond to  
$k$-adaptive adversaries.                                               

Another possible adversary definition could be to consider      
a sliding window of size $k$ as the adversary memory. But,       
because such an adversary can play actions to remind herself of events  
in the arbitrary past, her memory is not actually bounded by $k$,       
and regret minimization is not possible. See Appendix \ref{apx:StochasticGames} for more discussion.


$\KAdaptiveDefenderSet$ and $\KAdaptiveAdversarySet{K}$ denote all possible $K$-adaptive strategies for the defender and adversary, respectively.

\subsection{$k$-Adaptive Regret}
Suppose that the defender $\Defender$ and the adversary $\Adversary$ have produced history $H$ in a game $G$          
lasting~$T$ rounds. Let $a^1,...,a^T$ denote the sequence of actions    
played by the adversary. In hindsight we can construct a hypothetical $k$-adaptive adversary $\KAdaptiveAdversary{k}$ as follows:
      \[\KAdaptiveAdversary{k}\left(H'\right) = \Adversary\left(H^{t-i};H_i'\right)  \ ,\]
where $t = |H'|$ and $i = t \mod \left(k+1\right)$. In other words, the hypothetical $k$-adaptive adversary replicates the plays the real adversary made in the actual game regardless of the strategy of the defender he is playing against, {\em except} for the last $i$~rounds under consideration where he adapts his strategy to the defender's actions in the same manner the real adversary would. \cut{There are two important special cases: 
(1) Hypothetical Oblivious Adversary $\left(\KAdaptiveAdversary{0} \right)$: The hypothetical oblivious adversary plays a fixed sequence of actions always, 
(2) Hypothetical (Fully) Adaptive Adversary $\left(\KAdaptiveAdversary{\infty} \right)$: The hypothetical fully adaptive adversary {\em is} identical to the real adversary.}

Abusing notation slightly we write $\Pay\left(f,\Adversary,G, \sigma_0, T \right)$ to denote the expected payoff the defender would receive over $T$ rounds of $G$ given that the defender plays strategy $f$, the adversary uses strategy $\Adversary$ and the initial state of the bounded-memory game $G$ is $\sigma_0$. We use $ \AveragePay\left(f,\Adversary,G, T\right) =   \Pay\left(f,\Adversary,G, \sigma_0, T \right)/T$ to denote the average per-round payoff. We use \[\AverageRegret{k}\left(\Defender,\Adversary,G,T,S \right) = \max_{f \in S} \AveragePay\left(f,\KAdaptiveAdversary{k},G,T\right) - \AveragePay\left(\Defender, \KAdaptiveAdversary{k},G,T\right) \ , \]
to denote the {\em $k$-adaptive regret} of the defender strategy $\Defender$ using a fixed set $S$ of experts against an adversary strategy $\Adversary$ for $T$ rounds of the game $G$.

\begin{definition} \label{def:RegMinAlgApproximate}
A defender strategy $\Defender$ using a fixed set $S$ of experts is a {\em $\gamma$-approximate $k$-adaptive regret minimization algorithm } for the class of games $\GameClass$ if and only if for every adversary strategy $\Adversary$, every $\epsilon >0$ and every game $G \in \GameClass$ there exists $T' > 0$ such that $\forall T > T'$
\[\AverageRegret{k}\left(\Defender, \Adversary,G, T,S\right) < \epsilon + \gamma \ .\]
If $\gamma = 0$ then we simply refer to $\Defender$ as a k-adaptive regret minimization algorithm. If $\Defender$ runs in time $poly\left(n,1/\epsilon \right)$ we call $\Defender$ {\em efficient}. 
\end{definition}

$k$-adaptive regret considers a $k$-adaptive hypothetical adversary 
who can adapt within each window of size (at most) $k+1$. Intuitively,    
as $k$ increases this measure of regret is more meaningful (as the      
hypothetical adversary increasingly resembles the real adversary),      
albeit harder to minimize. 

There are two important special cases to consider: $k=0$ (oblivious regret) and $k=\infty$ (adaptive regret). Adaptive regret is the strongest measure of regret. Observe that if the actual adversary is $k$-adaptive then the hypothetical adversary $\KAdaptiveAdversary{\infty}$ is same as the hypothetical adversary $\KAdaptiveAdversary{k}$, and hence $\AverageRegret{\infty} = \AverageRegret{k}$. Also, if the actual adversary is oblivious then $\AverageRegret{\infty} = \AverageRegret{0} = \AverageRegret{k}$. 

In this paper $\GameClass$ will typically denote the class of perfect/imperfect information bounded-memory games with memory $m$. We are interested in expert sets $S$ which contain all of the fixed strategies $\FixedDefender \subseteq S$.


\section{Audit Examples} \label{sec:Audit}
As an example, consider the interaction between a police chief (defender) and drivers (adversary) at a popular tourist destination. The police chief is given the task of enforcing speed limits on local roads. Each day the police chief may deploy resources (e.g., radar, policemen) to monitor local roads, and drivers decide whether or not to speed or not. 

\paragraph{Repeated Game}
We first model the interaction above using a repeated game. We will consider a simple version of this interaction in which the defender has two actions \[\DefenderActions =  \left\{ \mathbf{HI}, \mathbf{LI} \right\} \ , \] and the adversary has two actions 
\[ \AdversaryActions = \left\{ \mathbf{S}, \mathbf{DS} \right\} \ . \]
Here,  {\bf HI}/{\bf LI} stands for high/low inspection and {\bf S}/{\bf DS} stands for speed and don't speed. We consider the defender utilities in table \ref{tab:speedingDefenderUtility}.

\begin{wraptable}{r}{3cm}
\centering
\begin{tabular}{| c | c | c |}
\hline
Actions & {\bf S} & {\bf DS} \\
\hline
{\bf HI} & .19 & 0.7 \\
\hline
{\bf LI} & 0.2 & 1 \\
\hline
\end{tabular}%
\caption{Speeding Game --- Defender Utility $\Pay$}%
\label{tab:speedingDefenderUtility}
\end{wraptable}

In this example, the costs of a higher inspection outweigh the benefits of enforcing the policy. In {\em any} Nash Equilibria the defender will play his dominant strategy  --- ``always play {\bf LI}." Similarly, {\em any} algorithm that minimizes regret in the standard sense (0-adaptive) --- like the regret minimizing audit mechanism from \cite{blocki2011regret} --- must eventually converge to the dominant defender strategy {\bf LI}. While this is the best that the defender can do against a byzantine adversary, this may not always be the best result for the defender when playing against a rational adversary. Consider the adversary's utility defined in table \ref{tab:speedingAdversaryUtility}.

\begin{wraptable}{r}{3cm}
\centering
\begin{tabular}{| c | c | c |}
\hline
Actions & {\bf S} & {\bf DS} \\
\hline
{\bf HI} & 0 & 0.8 \\
\hline
{\bf LI} & 1 & 0.8 \\
\hline
\end{tabular}
\caption{Speeding Game --- Adversary Utility }
\label{tab:speedingAdversaryUtility}
\end{wraptable}

If the defender plays his dominant strategy then the adversary will always play the action {\bf S} --- speed. This action profile results in average utility $0.2$ for the defender and $1$ for the adversary. However, if the defender can commit to his strategy in advance then he can play his Stackelberg equilibrium \cite{von2011market} strategy ``play {\bf HI} with probability $0.2$ and {\bf LI} with probability $0.8$." A rational adversary will respond by playing her best response --- the action that maximizes her utility given the defenders commitment. In this case the adversary's best response is to play {\bf DS}. The resulting utility for the defender is $0.94$! 

There are two practical challenges with adopting this approach: (1) If the utility of the adversary is unknown then the defender cannot compute the Stackelberg equilibrium. (2) Even if the defender commits to playing a Stackelberg equilibrium it is unlikely that many drivers will respond in purely rational manner for the simple reason that they are uniformed (e.g., a tourist may not know whether or not speed limits are aggressively enforce in an unfamiliar area). If the adversary can learn the Stackelberg Equilibrium from a history of the defender's actions, then she might adapt her play to the best response strategy over time. However, each tourist has a limited time window in which she can make these observations and adjust her behavior (e.g., the tourist leaves after at most $k$ days). 

\paragraph{Bounded Memory Game Model with $k$-adaptive regret}
We model the interaction above using bounded memory games with k-adaptive adversary model. In each round of our bounded memory game the defender and the adversary play an action profile, and observe an outcome --- a public signal. The action space in our bounded memory game is identical to the repeated game, and the outcome $\OutcomeSet=\{ \mathbf{HI}, \mathbf{LI}\}$ is simply the defender's action. That is we assume that our tourist driver can observe the defender's inspection level in each round (e.g., by counting the number of police cars by the side of the road). The defender's payoff function is identical to table \ref{tab:speedingDefenderUtility} --- the defender's payoff is independent of the current state (e.g., rewards in this particular bounded memory game are not history-dependent). A $k$-adaptive regret minimization algorithm could be run without a priori knowledge of the adversary's utility, and will converge to the optimal fixed strategy against any $k$-adaptive adversary (e.g., any sequence of $k$-adaptive tourist strategies). 

It is reasonable to use a $k$-adaptive strategy to model the behavior of our tourist drivers. Each tourist initially has no history of the defender's actions --- during the first day of her visit a tourist must make the decision about whether or not to speed without any history of the defender's actions. After the first day the tourist may adapt his behavior based on previous outcomes. For example, a tourist might adopt the following $k$-adaptive strategy: $\mathcal{A}_1 = $ ``Play {\bf DS} on the first day, and on the remaining $(k-1)$ days play {\bf S} if the defender has never played {\bf HI} previously, otherwise play {\bf DS}." After $k$ days the tourist leaves and a new tourist arrives. This new tourist may adopt a different $k$-adaptive strategy (e.g., $\mathcal{A}_2=$ ``Play {\bf S} on the first day, and on the remaining $(k-1)$ days play {\bf S} if the defender has never played {\bf HI} previously, otherwise play {\bf DS}.").

We set the memory of our bounded memory game to be $m=k$. Now the fixed defender strategies $\FixedDefender$ in our bounded memory game include strategies like $f=$ ``play $\mathbf{HI}$ every $k$'th round". Suppose for example that $k=7$ and the defender plays $f$. In this case the sequence of rewards that the defender would see against the first $k$-adaptive adversary $\mathcal{A}_1$ would be $(0.7,1,1,1,1,1,1)$. The sequence of rewards that the defender would see against the second $k$-adaptive adversary $\mathcal{A}_2$ would be $(0.19, 1, 1, 1,1,1,1)$. It is easy to verify that this is the optimal result for the defender --- if the defender does not play ${\bf HI}$ on the first day then the 7-adaptive adversary will speed on day 2. A $k$-adaptive regret minimization algorithm could be run without a priori knowledge of the adversary's utility, and will converge to the optimal fixed strategy against any $k$-adaptive adversary (e.g., any sequence of $k$-adaptive tourist strategies). 

\begin{remark}
A $k$-adaptive adversary is also an appropriate model for a temporary employee at the hospital so we could also consider the interaction between a hospital (defender) and a resident (adversary) at the hospital. The actions {\bf S} and {\bf DS}(e.g., ``speed" and ``don't speed") would be replaced with corresponding actions {\bf B} and {\bf V} (e.g., ``behave" and ``violate"). 
\end{remark}

Unfortunateley, we are able to prove that there is no efficient $k$-adaptive regret minimization algorithm for general bounded memory games. However, our results do not rule out the posibility of an efficient $\gamma$-approximate $k$-adaptive regret minimization algorithm. Finding an efficient $\gamma$-approximate $k$-adaptive regret minimization algorithms is an important open problem.

\section{Hardness Results} \label{sec:hardness}
In this section, we show that unless {\sf NP} = {\sf RP} no oblivious regret minimization algorithm which uses the fixed strategies $F$ as experts can be efficient in the imperfect information setting. In Appendix \ref{apdx:Hardness} we explain how our hardness reduction can be adapted to prove that there is no efficient $k$-adaptive regret minimization algorithm in the perfect information setting for $k \geq 1$.

Specifically, we consider the subclass of bounded-memory games $\BoundedMemoryGamesImperfectInfo$ with the following properties:
$\left|\mathcal{O} \right| = O(1)$,  $m = O\left(\log n\right)$, $\left|\AdversaryActions\right| = O(1)$, $\left| \DefenderActions\right| = O(1)$ and imperfect information. Any $G \in \BoundedMemoryGamesImperfectInfo$ is a game of imperfect information (on round $t$ the defender observes $O^t,$ but not $a^t$) with $O(n)$ states. Our goal is to prove the following theorem:
\newcommand{\thmHardness}{For any $\beta > 0$ and $\gamma < 1/8n^\beta$ there is no efficient 
$\gamma$-approximate oblivious regret minimization algorithm which uses the fixed strategies $\FixedDefender$ as experts against oblivious adversaries for the class of imperfect information bounded-memory-$m$ 
games unless {\sf NP} $=$ {\sf RP}. }

\begin{theorem} \label{thm:hardness}
\thmHardness
\end{theorem}

Given a slightly stronger complexity-theoretic assumption called the randomized exponential time hypothesis \cite{impagliazzo2001complexity} we can prove a slightly stronger hardness result. The randomized exponential time hypothesis says that no randomized algorithm running in time $2^{o(n)}$ can solve SAT.

\newcommand{\thmHardnessExpTimeHyp}{Assume that the randomized exponential time hypothesis is true. Then for any $\gamma < 1/\left(8 \log^2 n \right)$ there is no efficient $\gamma$-approximate oblivious regret minimization algorithm which uses the fixed strategies $\FixedDefender$ as experts against oblivious adversaries for the class of imperfect information bounded-memory-$m$ games.}
\begin{theorem} \label{thm:hardnessExpTimeHyp}
\thmHardnessExpTimeHyp
\end{theorem}

\newcommand{\lemmaHardness}{Unless {\sf NP} $=$ {\sf RP}, for $\gamma < 1/8n$ 
there is no efficient $\gamma$-approximate oblivious regret minimization algorithm which uses the fixed strategies $\FixedDefender$ as experts 
against oblivious adversaries for bounded-memory-$m$ games of imperfect information.}

The proofs of Theorems \ref{thm:hardness} and \ref{thm:hardnessExpTimeHyp} use the fact that it is hard to approximate MAX3SAT within any factor better than $\frac{7}{8}$ \cite{hastad2001some}. This means that unless {\sf NP} $=$ {\sf RP} then for every constant $\beta > 0$ and every randomized algorithm $S$ in $RP$, there exists a MAX3SAT instance $\phi$ such that the expected number of clauses in $\phi$ unsatisfied by $S(\phi)$ is $\geq \frac{1}{8} - \beta$ even though there exists an assignment satisfying $(1-\beta)$ fraction of the clauses in $\phi$.

We reduce a MAX3SAT formula $\phi$ with variables $x_1,...,x_n$ and clauses $C_1,...,C_\ell$ to a bounded-memory game $G$ described formally below.
We provide a high level overview of the game $G$ before describing the details. The main idea is to construct $G$ so that the rewards in $G$ are related to the fraction of clauses of $\phi$ that are satisfied.

In $G$, for each variable $x$ there is a state $\sigma_x$ associated with that variable. The oblivious adversary controls the transitions between variables. This allows the oblivious adversary $\Adversary_R$ to partition the game into stages of length $n$, such that during each stage the adversary causes the game to visit each variable exactly once (each state is associated with a variable). During each stage the adversary picks a clause $C$ at random. In $G$ we have $0,1 \in \DefenderActions$.  Intuitively, the defender chooses assignment $x=1$ by playing the action $1$ while visiting the variable $x$.  The defender receives a reward if and only if he succeeds in satisfying the clause $C$.

The game $G$ is defined as follows:\\
\noindent {\bf Defender Actions: } $\DefenderActions = \{0,1,2\}$ \\
\noindent {\bf Adversary Actions:} $\AdversaryActions = \{0,1\}\times \{0,1,2,3 \}$ \\
\noindent {\bf Outcomes and States: }
Each round $i$ produces two outcomes \[\tilde{ O}^i = \vec{a}^i[1] \mbox{ ~~~and~~~ } \hat{O}^i = \begin{cases}
1 & \text{if $d^i = 2$ or $d^i = a^i[2]$}; \\
0 & \text{otherwise}.
\end{cases} \] Observe that these outcomes satisfy the independent outcomes requirement for bounded-memory games. There are $n = 2^{m+1}$ states, where $\sigma^i$ is the state at round $i$, where \[ \sigma^i = \left(\langle \tilde{O}^{i-1}, \ldots ,\tilde{O}^{i-m}\rangle, \hat{O}^{i-1} \right) \ . \]   Observe that each state
encodes the last $m$ outcomes $\tilde{O}$ and the last outcome
$\hat{O}^i$. Intuitively, the last $m$ outcomes $\tilde{O}^i$ are used
to denote the variable $x_i$, while $\hat{O}^i$ is $1$ if the defender
has already received a reward during the current phase.

The defender actions $0,1$ correspond to the truth assignments $0,1$. The defender receives a reward for the correct assignment. The defender is punished if he attempts to obtain a reward in any phase after he has already received a reward in that phase. Once the defender has already received a reward he can play the special action $2$ to avoid getting punished. The intuitive meaning of the adversary's actions is explained below.

If we ignore the outcome $\hat{O}$ then the states form a De Bruijn graph \cite{good1946normal} where each node corresponds to a variable of $\phi$. Notice that the adversary completely controls the outcomes $\tilde{O}$ with the first component of his action $\vec{a}[1]$. By playing a De Bruijn sequence $S = s_1...s_n$ the adversary can guarantee that we repeatedly take a Hamiltonian cycle over states(for an example see Figure \ref{redVariableTransitions}). 

\begin{figure}

\begin{center}
\includegraphics[width=0.6\textwidth]{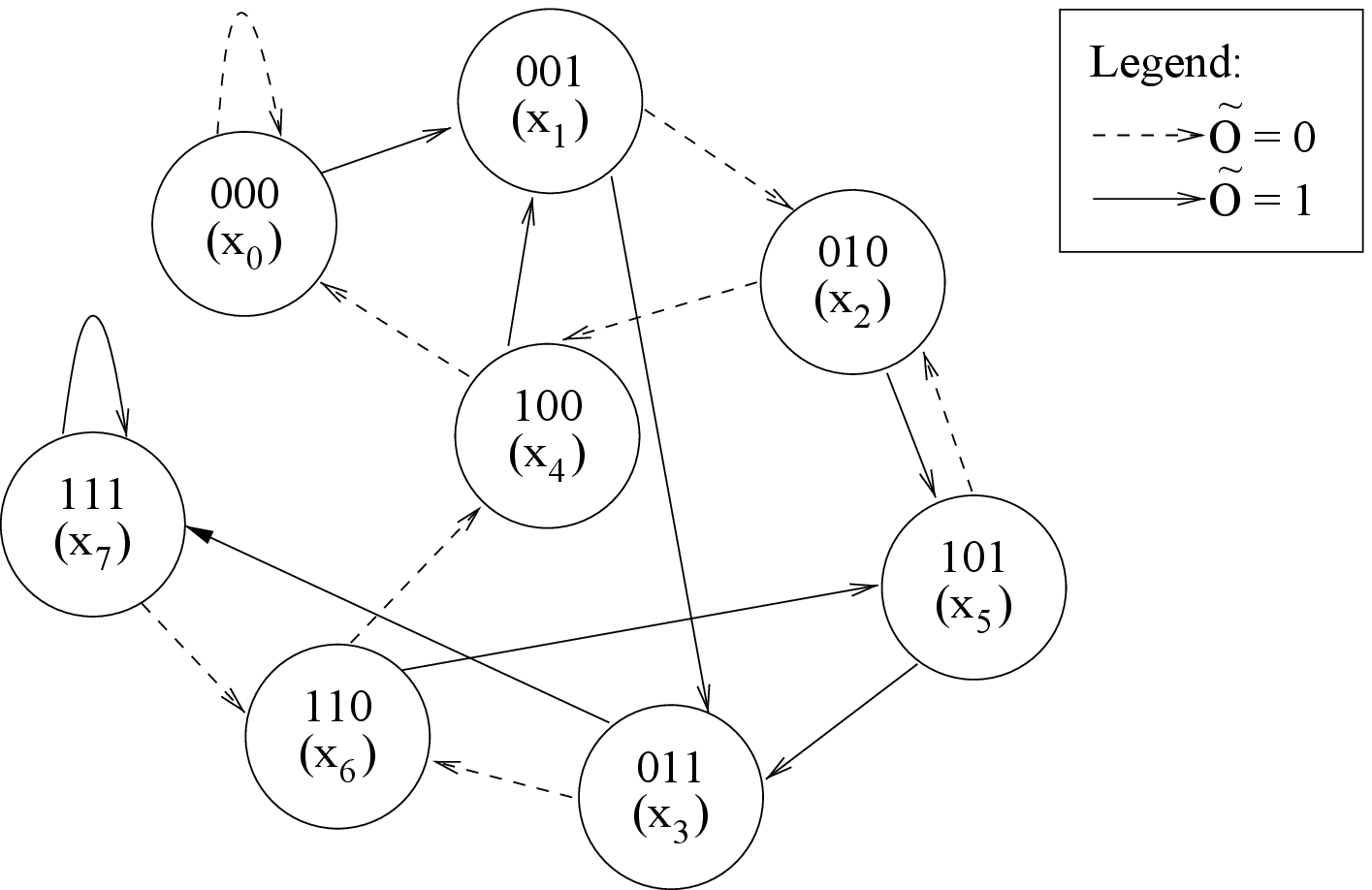}
\end{center}
\caption{De Bruijn example}
\label{redVariableTransitions}

\end{figure}

\smallskip
\noindent {\bf Rewards:\footnote{We use payoffs in the range $[-1,1]$ for ease of presentation. These payoffs can easily be re-scaled to lie in $[0,1]$.}}
\[ 
\Pay\left(\sigma^i,d^i,a^i\right) = \begin{cases}
-1 & \text{if $\hat{O}^{i-1} = 1$ and $d^i \neq 2$ and $\vec{a}^i[2] \neq 3$}; \\
1 & \text{if $d^i \neq 2$ and $d^i = \vec{a}^i[2]$ and $\hat{O}^{i-1} = 0$}; \\
0 & \text{otherwise}.
\end{cases}
\]
An intuitive interpretation of the reward function is presented in parallel with the adversary strategy. 

\noindent {\bf Adversary Strategy:} The first component of the adversary's action ($\vec{a}[1]$) controls the transitions between variables. The adversary will play the action $\vec{a}^i[2] = 1$ (resp. $\vec{a}^i[2] = 0$) whenever the corresponding variable assignment $x_i=1$ (resp. $x_i = 0$) satisfies the clause that the adversary chose for the current phase. \\

 If neither variable assignment satisfies the clause (if $x_i \notin C$ and $\bar{x}_i \notin C$) then the adversary plays $\vec{a}^i[2] = 2$. This ensures that a defender can only be rewarded during a round if he satisfies the clause $C$, which happens when $d^i = \vec{a}^i[2] = 0 \mbox{ or } 1$.

Notice that whenever $\hat{O} = 1$ there is no way to receive a positive reward. The defender may want the game $G$ to return to a state where $\hat{O} = 0$, but unless the adversary plays the special action $\vec{a}^i[2] = 3$ he is penalized when this happens. The adversary action $\vec{a}^i[2] = 3$ is a special `reset phase' action. By playing $\vec{a}^i[2]= 3$ once at the end of each phase the adversary can ensure that the maximum payoff the defender receives during any phase is $1$. See figure \ref{alg:ObliviousAdversary} for a formal description of the adversary strategy.

\begin{figure}
\noindent $\bullet$ {\bf Input:} MAX3SAT instance $\phi$, with variables $x_1,\ldots, x_{n-1} $ ,  and clauses $C_1, \ldots , C_\ell $. Random string $R \in \{0,1\}^*$

\noindent$\bullet$    {\bf  De Bruijn sequence: } $s_0,...,s_{n-1}$

\noindent$\bullet$   {\bf Round $t$: } Set $i \leftarrow t \mod n$.

\noindent\quad \; 1.  {\bf Select Clause:}  If $i = 0$ then select a clause $C$ uniformly at random from $C_1,...,C_\ell$ using $R$. 

\noindent\quad \; 2. {\bf Select Move: } 
\[ 
a^i  = \begin{cases}
(s_i,3) & \text{if $i = 0$}; \\
(s_i,1) & \text{if $x_i \in C$};\\
(s_i,0) & \text{if $\bar{x}_i \in C$}; \\
(s_i,2) & \text{otherwise}.
\end{cases}
\]
\caption{Oblivious Adversary: $\Adversary_R$}
\label{alg:ObliviousAdversary}
\end{figure}


\noindent  {\bf Analysis: } At a high level, our hardness argument proceeds as follows: 
\begin{enumerate}
\item If there is an assignment that satisfies $(1-\beta)$ fraction of the clauses in $\phi$, 
then there is a fixed strategy that performs well in expectation (see Claim \ref{claim:hardnessFixedStrategy}).
\item  If there a fixed strategy that performs well in expectation, then any $\gamma$-approximate oblivious regret 
minimization algorithm will perform well in expectation (see Claim \ref{claim:hardnessDefenderPerfomance}).
\item  If an efficiently computable strategy $\Defender$ performs well in expectation, then there is an efficiently computable randomized algorithm $S$ to approximate MAX3SAT (see Claim \ref{claim:recoverAssignment}). This would imply that {\sf NP} $=$ {\sf RP}.

\end{enumerate}
\medskip

\newcommand{\claimHardnessFixedStrategy}{Suppose that there is a variable assignment that satisfies $\left(1-\beta\right)\cdot \ell$ of the clauses in $\phi$.  Then there is a fixed strategy $f$ such that $ E_R \left[\AveragePay\left(f,\Adversary_R,G,n \right) \right] \geq \left(1-\beta\right)/n$ ,
where $R$ is used to denote the random coin tosses of the oblivious adversary.}

\begin{claimN} \label{claim:hardnessFixedStrategy}
\claimHardnessFixedStrategy 
\end{claimN}

\newcommand{\claimHardnessDefenderPerformance}{Suppose that $\Defender$ is an $\left(\frac{1}{8n} - \frac{3\beta}{n} \right)$-approximate oblivious regret minimization algorithm 
against the class of oblivious adversaries and there is a variable assignment that satisfies $(1-\beta)$ fraction of the clauses in 
$\phi$. Then for $T = poly(n)$
\[E_R \left[\AveragePay\left(\Defender,\Adversary_R,G,T \right) \right] \geq \frac{7}{8n} + \frac{\beta}{n} \ , \]
where $R$ is used to denote the random coin tosses of the oblivious adversary.}
 
\begin{claimN} \label{claim:hardnessDefenderPerfomance}
\claimHardnessDefenderPerformance
\end{claimN}

\newcommand{\claimRecoverAssignment}{Fix a polynomial $p(\cdot)$ and let
$\alpha = n \cdot E_R \left[\AveragePay\left(\Defender,\Adversary_R,G,T \right) \right]$  ,
where $T = p(n)$ and $\Defender$ is any polynomial time computable strategy. There is a polynomial time randomized algorithm $S$ which 
satisfies $\alpha$ fraction of the clauses from $\phi$ in expectation.  }

\begin{claimN} \label{claim:recoverAssignment}
\claimRecoverAssignment
\end{claimN}

The proofs of these claims can be found in Appendix \ref{apdx:Hardness}. \\
\begin{proofof}{Theorem \ref{thm:hardness}}
The key point is that if an algorithm $S$ runs in time $O\left(p(n)\right)$ on instances of size $n^\beta$ for some polynomial $p(n)$ then on instances of size $n$ $S$ runs in time $O\left(p\left(n^{1/\beta} \right) \right)$ which is still polynomial time. Unless {\sf NP} $=$ {\sf RP} $\forall \epsilon, \beta > 0$ and every algorithm $S$ running in time poly(n), there exists an integer $n$ and a MAX3SAT formula $\phi$ with $n^\beta$ variables such that 
\begin{enumerate}
\item There is an assignment satisfying at least $(1-\epsilon)$ of the clauses in $\phi$.
\item The expected fraction of clauses in $\phi$ satisfied by $S$ is  $\leq \frac{7}{8} + \epsilon$. 
\end{enumerate}
If we reduce from a MAX3SAT instance with $n^\beta$ variables we can construct a game with $O(n)$ states ($n^{1-\beta} $ copies of each variable). One Hamiltonian cycle would now corresponds to $n^{1-\beta} $ phases of the game. This means that the expected average reward of the optimal fixed strategy is at least
\[\max_{f \in \FixedDefender} E_R\left[ \AveragePay\left(f, \Adversary_R, G,T \right)\right] \geq \frac{n^{1-\beta} \left(1-\epsilon\right)}{n}  \ ,\]
while the expected average reward of an efficient defender strategy $\Defender$ is at most
\[   E_R\left[ \AveragePay\left(\Defender, \Adversary_R, G,T \right)\right] \leq \frac{n^{1-\beta} \left(\frac{7}{8}+\epsilon\right)}{n} \ .\]
Therefore, the expected average regret is at least
\[\AverageRegret{0}\left(\Defender, \Adversary_R, G,T,\FixedDefender \right) \geq \left(\frac{1}{8}-2\epsilon \right)n^{-\beta} \ . \]
\end{proofof}

 The proof of theorem \ref{thm:hardnessExpTimeHyp} is similar to the proof of theorem \ref{thm:hardness}. It can be found in Appendix \ref{apdx:Hardness}.

\newcommand{\proofofLemmaHardness}{
\begin{proofof}{Lemma~\ref{lemma:hardness}}
Suppose that $\Defender$ were an efficient $\gamma$-approximate oblivious regret minimization algorithm and consider the 
polynomial time randomized algorithm $S$.  Combining Claim \ref{claim:hardnessDefenderPerfomance} and Claim
\ref{claim:recoverAssignment}, for every MAX3SAT formula $\phi$ with $\geq (1-\beta)$ fraction of the clauses 
satisfiable $S$ satisfies $\geq \frac{7}{8} + \beta$ fraction of the clauses from $\phi$ in expectation.  
This would imply that {\sf NP} $=$ {\sf RP} \cite{hastad2001some}.
\end{proofof}}


Our hardness reduction is similar to a result from Even-Dar, et al.,                       
\cite{even2005experts}. They consider regret minimization in a Markov   
Decision Process where the adversary controls the transition model.     
Their game is not a bounded-memory game; in particular it does not      
satisfy our \emph{independent outcomes} condition. The current state in their game can depend    
on the last $n$ actions. In contrast, we consider bounded-memory games  
with $m = O\left(\log n\right)$, so that the current state only depends 
on the last $m$ actions. This makes it much more challenging to enforce 
guarantees such as ``the defender can only receive a reward once in     
each window of $n$ rounds''---a property that is used in the hardness   
proof. The adversary is oblivious so she will not remember this fact,   
and the game itself cannot record whether a reward was given $m+1$      
rounds ago. We circumvented this problem by designing a payoff function 
in which the defender is penalized for allowing the game to ``forget''  
when the last reward was given, thus effectively enforcing the desired  
property.                                                               


%
%
\section{Regret Minimization Algorithms} \label{sec:algorithm}
In section \ref{subsec:ReductionToRepeatedGames} we present a reduction from bounded-memory games to repeated games. This reduction can be used to create a $k$-adaptive regret minimizing algorithm (Theorem \ref{thm:memoryDTIRRegretMin}). This is significant because there is no $k$-adaptive regret minimization algorithm for the general class of stochastic games. A consequence of Theorem \ref{thm:hardness} is that when the expert set includes all fixed strategies $\FixedDefender$ we cannot hope for an efficient algorithm unless {\sf NP} $=$ {\sf RP}. In section \ref{subsec:efficient} we present an efficient {\em approximate} $0$-adaptive regret minimization algorithm for bounded-memory games of perfect information. The algorithm uses an implicit weight representation to efficiently sample the experts and update their weights. Finaly, we show how this algorithm can be adapted to obtain an efficient approximate $0$-adaptive regret minimization algorithm for bounded-memory games of {\em imperfect} information.

\subsection{Reduction to Repeated Games} \label{subsec:ReductionToRepeatedGames}
All of our regret minimization algorithms work by first reducing the bounded-memory game $G$ to a repeated game $\RepeatedGame{G}{K}$. One round of the repeated game $\RepeatedGame{G}{K}$ corresponds to $K$ rounds of $G$. Before each round of $\RepeatedGame{G}{K}$ both players commit to an adaptive strategy. In $\RepeatedGame{G}{K}$ the reward that the defender gets for playing a strategy $f \in \KAdaptiveDefenderSet$ is the reward that the defender would have received for using the strategy $f$ for the next $K$ rounds of the actual game $G$ if the initial state were $\sigma_0$: $\Pay\left(f,g,\RepeatedGame{G}{K}\right) = \Pay\left(f,g, G, \sigma_0, K\right)$.

The rewards in $\RepeatedGame{G}{K}$ may be different from the actual rewards in $G$ because the initial state before each $K$ rounds might not be $\sigma_0$. Claim \ref{claim:LossDifferences} bounds the difference between the hypothetical losses from $\RepeatedGame{G}{K}$ and actual losses in $G$ using the bounded-memory property. The proof of Claim \ref{claim:LossDifferences} is in Appendix \ref{apdx:RegretMinAlg}.

\newcommand{\claimLossDifferences}{For any adaptive defender strategy $f \in \KAdaptiveDefenderSet$ and any adaptive adversary strategy $g \in \KAdaptiveAdversarySet{K}$ and any state $\sigma$ of $G$ we have 
$\left| \Pay\left(f,g, G , \sigma, K \right) - \Pay\left(f,g, G, \sigma_0, K\right) \right| \leq m $ . 
 }

\begin{claimN} \label{claim:LossDifferences}
\claimLossDifferences
\end{claimN}

The key idea behind our $k$-adaptive regret minimization algorithm  $\BndMemRegMin$ is  
to reduce the original bounded-memory game to a repeated game $\RepeatedGame{G}{K}$ of      
imperfect information ($K \equiv 0 \mod{k}$). In particular we obtain the regret bound in Theorem \ref{thm:memoryDTIRRegretMin}. Details and proofs can be found in Appendix \ref{apdx:RegretMinAlg}.

\newcommand{\ImpInfoRegMinAlgThm}{Let $G$ be any bounded-memory-$m$ game with $n$ states and let $\Adversary$ be any adversary strategy. After playing $T$ rounds of $G$ against $\Adversary$,  $\BndMemRegMin\left(G,K\right)$ achieves regret bound \begin{eqnarray*}
 \AverageRegret{k}\left(\BndMemRegMin, \Adversary,G,T,S\right) & < &  \frac{m}{T^{1/4}} +  4\frac{\sqrt{  N \log N}}{T^{1/4}} \ ,
\end{eqnarray*}
where $N = \left|S\right|$ is the number of experts, $\Adversary$ is the adversary strategy and $K$ has been chosen  so that $K =  T^{1/4}$ and $K \equiv 0 \mod{k}$.}
\begin{theorem} \label{thm:memoryDTIRRegretMin} 
\ImpInfoRegMinAlgThm
\end{theorem}
Intuitively, the $m/T^{1/4} = m/K$ term is due to modeling loss from Claim \ref{claim:LossDifferences} and the other term comes from the standard regret bound of \cite{auer1995gambling}.

\subsection{Efficient Approximate Regret Minimization Algorithms} \label{subsec:efficient}
In this section we present $\RegMinPerfInfo$ (Efficient approXimate Bounded Memory Weighted Majority), an efficient algorithm to approximately minimize  regret against an oblivious adversary in bounded-memory games with perfect information. The set of experts $\PerfectInfoExpertSet$ used by our algorithms contains the fixed strategies $\FixedDefender$ as well as all $K$-adaptive strategies $\KAdaptiveDefenderSet$ ($K=m/\gamma$).  We prove the following theorem

\newcommand{\RegMinPerfInfoThm}{Let $G$ be any bounded-memory-$m$ game  of perfect information with $n$ states and let $\Adversary$ be any adversary strategy. Playing $T$ rounds of $G$ against $\Adversary$, $\RegMinPerfInfo$ runs in total time $Tn^{O\left(1/\gamma\right)}$ and achieves regret bound
\[\AverageRegret{0}\left(\RegMinPerfInfo, \Adversary,G, T,\PerfectInfoExpertSet \right) \leq \gamma + O\left(\frac{m}{\gamma}\sqrt{\frac{\frac{m}{\gamma} n \log \left(N   \right)}{T}}\right) \ , \]
where $K$ has been set to $m/\gamma$ and $N = \left|\KAdaptiveDefenderSet \right| =  \left(\left|\DefenderActions\right|\right)^{n^{1/\gamma}}$ is the number of $K$-adaptive strategies.}

\begin{theorem}\label{thm:PerfectInformation}
\RegMinPerfInfoThm
\end{theorem}

In particular, for any constant $\gamma$ there is an efficient $\gamma$-approximate $0$-adaptive regret minimization algorithm for bounded-memory games of perfect information. We can adapt this algorithm to get $\RegMinImperfInfo$ (Efficient approXimate Bounded Memory Weighted Majority for Imperfect Information Games), an efficient approximate 0-adaptive regret minimization algorithm for games of imperfect information using a sampling strategy described in Appendix \ref{apdx:RegretMinAlg}. 

\begin{theorem}\label{thm:ImperfectInformationApproximateRegret}
Let $G$ be any bounded-memory-$m$ game of imperfect information with $n$ states and let $\Adversary$ be any adversary strategy. There is an algorithm $\RegMinImperfInfo$ that runs in total time $Tn^{O\left(1/\gamma\right)}$ playing $T$ rounds of $G$ against $\Adversary$, and achieves regret bound
\[\AverageRegret{0}\left(\RegMinImperfInfo, \Adversary,G, T, \PerfectInfoExpertSet \right) \leq 2\gamma + O\left(\frac{m n^{1/\gamma}}{\gamma^2}\sqrt{\frac{\frac{mn^{1/\gamma}}{\gamma} n \log \left(N   \right)}{T}}\right) \ . \]
where $K$ has been set to $m/\gamma$ and $N = \left|\KAdaptiveDefenderSet \right| =  \left(\left|\DefenderActions\right|\right)^{n^{1/\gamma}}$ is the number of $K$-adaptive strategies.
\end{theorem}

The regret bound of Theorem \ref{thm:PerfectInformation} is simply the regret bound achieved by the standard weighted majority algorithm \cite{littlestone1989weighted} plus the modeling loss term from Claim \ref{claim:LossDifferences}. The main challenge is to provide an efficient simulation of the weighted majority algorithm. There are an exponential number of experts so no efficient algorithm can explicitly maintain weights for each of these experts. To simulate the weighted majority algorithm $\RegMinPerfInfo$ implicitly maintains the weight of each expert. 

To simulate the weighted majority algorithm we must be able to {\em efficiently sample} from our weighted set of experts (see $\SampleAlgorithm$) and efficiently update the weights of each expert in the set after each round of $\RepeatedGame{G}{K}$ (see update weight stage of $\RegMinPerfInfo$). 

\noindent {\bf Meet the Experts} Instead of using $\FixedDefender$ as the set of experts, $\RegMinPerfInfo$ uses a larger set of experts $ \PerfectInfoExpertSet$ ($\FixedDefender \subset \PerfectInfoExpertSet$). Recall that a $K$-adaptive strategy is a function $f$ mapping the $K$ most recent outcomes $H_K$ to actions. We use a set of $K$-adaptive strategies $E= \{f_\sigma : \sigma \in \Sigma \} \subset \KAdaptiveDefenderSet$ to define an expert $E$ in $\RepeatedGame{G}{K}$: if the current state of the real bounded-memory game $G$ is $\sigma$ then $E$ uses the $K$-adaptive strategy $f_\sigma$ in the next round of $\RepeatedGame{G}{K}$ (i.e., the next $K$~rounds of $G$).  $\PerfectInfoExpertSet$ denotes the set of all such experts. 

\noindent {\bf Maintaining Weights for Experts Implicitly} To implicitly maintain the weights of each expert $E \in \PerfectInfoExpertSet$ we use the concept of a game trace. We say that a game trace $p = \sigma, d^1,O^1,...,d^{i-1},O^{i-1},d^i$ is consistent with an expert $E$ if $f_\sigma\left(O^1,...,O^{j-1} \right) = d^j$ for each $j$. We define the set $\Consistent{E}$ to be the set of all such consistent traces of maximum length $K$ and $\ConsistentAll = \bigcup_{E \in \PerfectInfoExpertSet} \Consistent{E}$ denotes the set of all traces consistent with some expert $E \in \PerfectInfoExpertSet$. $\RegMinPerfInfo $ maintains a weight $w_p$ on each trace $p \in \ConsistentAll$. The weight of an expert $E$ is then defined to be
$W_E = \prod_{p \in \Consistent{E}} w_p$. 

Given adversary actions $\vec{a} = a_1,...,a_K$ and a trace $p = \sigma, d^1,O^1,...,d^{i-1},O^{i-1},d^i$ we define $\ReachableTrace{\vec{a},\sigma ', p}$.\\
\begin{wrapfigure}{r}{0.7 \textwidth}
\vspace{-30pt}
 \[\ReachableTrace{\vec{a},\sigma ', p} =  \begin{cases}
0 & \text{if $\sigma \neq \sigma '$}; \\
\prod_{j<i} \Pr\left[O^j ~\vline ~ a^j, d^j \right] & \text{otherwise};\\
\end{cases} \]
\vspace{-10pt}
\end{wrapfigure}
Intuitively, $\ReachableTrace{\vec{a},\sigma ', p}$ is the probability that each outcome of $p$ would have occurred given the adversary actions were $\vec{a}$ and the initial state was $\sigma '$. We use $\ell\left(p,\vec{a},\sigma '\right)$ to denote the payment that the defender received for playing $d^i$ (the last action in $p$). Formally $
\ell\left(p,\vec{a}, \sigma '\right) = \Pay\left(\sigma^f_p,d^i, a^i \right) \ReachableTrace{\vec{a},\sigma ', p} $,
where $\sigma^f_p$ denotes the state reached following the trace $p$ (after observing outcomes $O^1,...,O^{i-1}$ starting from $\sigma_0$) and $d^i$ is the final defender action in the trace.  Notice that in the imperfect information setting the defender could not compute $\ell$ because he would not observe the adversary's actions $\vec{a}$. 

\noindent {\bf Updating Weights Efficiently} While updating weights $\RegMinPerfInfo$ maintains the invariant that $ w_p = \beta^{\sum_{j=1}^{T/K} \ell\left(p,\vec{a}^j,\sigma^{jK}\right)}$, 
where $\sigma^{jK}$ is the state of $G$ after $jK$ rounds and $\vec{a}^t$ is the actions the adversary played during the $j$'th round of $\RepeatedGame{G}{K}$. The standard weighted majority algorithm maintains the invariant that 
$W_E = \beta^{\sum_{j=1}^{T/K} \Pay\left(E, \vec{a}^t, \RepeatedGame{G}{K}  \right)} $. Claim \ref{claim:PerfInfoClaim} implies that $\RegMinPerfInfo$ also maintains this invariant with its implicit weight representation --- the proof of Claim \ref{claim:PerfInfoClaim} is in Appendix \ref{apdx:RegretMinAlg}. 
\begin{claimN} \label{claim:PerfInfoClaim}
\[
\prod_{p \in \Consistent{ E}} \beta^{\sum_{j=1}^{T/K} \ell\left(p,\vec{a}^j,\sigma^{jK}\right)} = \beta^{\sum_{j=1}^{T/K} \Pay\left(E, \vec{a}^j, \RepeatedGame{G}{K} \right)} \ . \]
\end{claimN}

\noindent {\bf Sampling Experts Efficiently} We can also efficiently sample from $\PerfectInfoExpertSet$ using dynamic programming (see $\SampleAlgorithm$). Using the notation $p \sqsubset p'$ for $p'$ extends $p$ we can define $\hat{w}_p$. Intuitively, $\hat{w}_{p;O;d}$ represents the weight of the action $d$ from history $p;O$. 
\begin{wrapfigure}{r}{0.4 \textwidth}
\vspace{-15pt}
\[ \hat{w}_p = \sum_{E:p \in \Consistent{E}} \prod_{p' \in \Consistent{E} \wedge p \sqsubset p' } w_{p'}   \]
\vspace{-20pt}
\end{wrapfigure}
Using dynamic programming we can efficiently compute $\hat{w}_p$ for each trace $p$ because there are only $n^{O\left(1/\gamma \right)}$ such traces. Using the weights $\hat{w}_p$ we can efficiently sample from $\PerfectInfoExpertSet$. We use $p;O;d$ to denote a new game trace which contains all of the outcomes/actions in $p$ appended with $O$ and $d$.

\begin{minipage}[t]{0.42\textwidth}
\hrule
{\bf Algorithm:} $\RegMinPerfInfo\left(\gamma, G\right)$
\hrule 
\noindent $\bullet$ {\bf Initialize:} $K = m/\gamma$ \newline 
\noindent $\bullet$  {\bf Construct: } $\RepeatedGame{G}{K}$ \newline
\noindent $\bullet$  {\bf Each Round:} \newline 
\cut{\begin{enumerate}
\item $\sigma \leftarrow G.CurrentState$
\item $E\leftarrow \SampleAlgorithm$
\item Play $E$
\item  Observe adversary actions $$\vec{a} = a^1, ..., a^K \ . $$
\item  {\bf Update Weights: } For each $p \in \ConsistentAll$ 
\end{enumerate}}
\noindent $~~~~~~$ 1. $\sigma \leftarrow G.CurrentState$ \newline 
\noindent $~~~~~~$ 2. $E\leftarrow \SampleAlgorithm$ \newline
\noindent $~~~~~~$ 3. Play $E$ \newline
\noindent $~~~~~~$ 4. Observe adversary actions $$\vec{a} = a^1, ..., a^K \ . $$ 
\noindent 5. {\bf Update Weights: } For each $p \in \ConsistentAll$ \newline 
$~~~~~~~~$ A. Compute $\ell\left(p,\vec{a},\sigma\right)$ \newline
$~~~~~~~~$ B. Set $w_p \leftarrow w_p \times \beta^{\ell\left(p,\vec{a},\sigma\right)}$. 
\hrule
\end{minipage}
\begin{minipage}[t]{0.53\textwidth}
\hrule
{\bf Algorithm:} $\SampleAlgorithm$
\hrule 
\noindent $\bullet$ For each trace $p \in \ConsistentAll$ recursively compute $\hat{w}_p$ using the formula:
   \[ \hat{w}_p = \sum_{O \in \OutcomeSet} \sum_{d \in \DefenderActions} \beta^{\sum_{t=1}^T \ell\left(p;O;d, \vec{a}^t, \sigma^{Kt} \right) } \hat{w}_{p;O;d} \ . \]
\noindent $\bullet$ {\bf Build Strategy} $E$: For each $p\in \ConsistentAll$ and   $O \in \OutcomeSet$, randomly select $d \in \DefenderActions$ 
\[\Pr\left[d ~\vline~ p,~ O \right] = \frac{\hat{w}_{p;O;d}}{\sum_{d' \in \DefenderActions} \hat{w}_{p;O;d'} } \ .\]
\noindent $\bullet ~E$ play $d$ any time it observes history $p;O$.

\hrule

\end{minipage}

Claim \ref{claim:Sample} says that $\SampleAlgorithm$ outputs each expert $E$ with probability proportional to $W_E$.
\begin{claimN} \label{claim:Sample}
For each expert $E \in \PerfectInfoExpertSet$ Algorithm $\SampleAlgorithm$ outputs $E$ with probability 
\[ \Pr\left[E \right] \propto W_E \ . \]
\end{claimN}
 Given $\SampleAlgorithm$ it is straightforward to simulate the standard weighted majority algorithm. To update weights $\RegMinPerfInfo$ simply loops through all traces $p \in \ConsistentAll$ applying the update rule
$w_p = w_p \times \beta^{\ell\left(p,\vec{a}^t, \sigma^{tK}\right)}$, where $\beta$ is a learning parameter we tune later. The formal proof of Theorem \ref{thm:PerfectInformation} can be found in Appendix \ref{apdx:RegretMinAlg} along with the proof of claim \ref{claim:Sample}.

At a high level our algorithm is similar to the online shortest path algorithm developed by Takimoto and Warmuth \cite{takimoto2003path}. In their work, they consider the set of all source-destination paths in a graph as experts. Since there are exponentially many paths they also maintain the weights of the experts implicitly. In their setting, the defender completely controls the chosen path. In contrast, our experts adapt to adversary actions. The challenge was constructing a new implicit weight representation which works for $K$-adaptive strategies. 

Using this implicit weight representation we could have also used the general  barycentric spanner approach to online linear optimization developed by  Awerbuch and Kleinberg \cite{awerbuch2008online} to design a $\gamma$-approximate $0$-adaptive regret minimization algorithm running in time $n^{O\left(1/\gamma\right)}$. However, we are able to achieve better regret bounds in theorem \ref{thm:PerfectInformation} by simulating the weighted majority algorithm. Awerbuch and Kleinberg \cite[Theorem 2.8]{awerbuch2008online} achieve the average regret bound $O\left(Md^{5/3}/T^{1/3} \right)$, where $d$ is the dimension of the problem space and $M$ is a bound on the cost vectors. By comparison our regret bounds in Theorems \ref{thm:PerfectInformation} and \ref{thm:ImperfectInformationApproximateRegret} tend to $0$ with $1/\sqrt{T}$. In our setting, the dimension of the problem space is $d = O\left(n^{\left(1/\gamma \right)}\right)$ (the number of nodes in the decision tree), and $M = K = m/\gamma$ is the upper bound on the cost vector in each round of $\RepeatedGame{G}{K}$. The average regret bound would be $O\left(\frac{m}{\gamma}n^{5/\left(3\gamma\right)}/T^{1/3} \right)$. 
the regret bound is proportional to $\sqrt{n^{1/\gamma}/T}$. By comparison Theorem \ref{thm:PerfectInformation} has a $\sqrt{n^{1/\gamma}}$ in the numerator.

The standard regret minimization trick for dealing with imperfect information in a repeated game is to break the game up into phases and perform random sampling in each round to estimate the cost of each expert and update weights. The challenge in adapting $\RegMinPerfInfo$ is that there are exponentially many experts in $\PerfectInfoExpertSet$. Our key idea was to estimate $\ell\left(p, \vec{a}, \sigma\right)$ for each $p \in \ConsistentAll$ so there are only $n^{O\left(1/\gamma \right)}$ samples to take in each phase. We can then update the implicit weight representation using the estimated values $\ell\left(p,\vec{a},\sigma\right)$.


\section{Open Questions} \label{sec:conclusion}
In this paper, we defined a new class of games called bounded-memory    
games, introduced several new notions of regret, and presented hardness
results and algorithms for  regret minimization in this subclass of stochastic games.                        
Because both the games and the notions of regret we study in this paper 
rely on novel definitions, they raise a number of interesting open  
problems:
(1) To what extent can the hardness results of Theorems \ref{thm:hardness} and \ref{thm:hardnessExpTimeHyp} be further improved?
($\gamma =1/{\log n}$?) Could similar hardness results apply to games with perfect information?            
(2) Is there an efficient {\em non-approximate} oblivious regret minimization algorithm for bounded-memory games with perfect information?
(3) Is there a $\gamma$-approximate oblivious regret minimization algorithm with running time $n^{o\left(1/\gamma \right)}$? For example, could one design a $\gamma$-approximate oblivious regret minimization algorithm with running time $n^{-\log \gamma}$? 
(4) For repeated games $(m=0)$ is there an efficient $\gamma$-approximate $k$-adaptive regret minimization algorithm if we use $\KAdaptiveDefenderSet$ as our set of experts $(K=\log n)$?

\bibliographystyle{splncs}
\bibliography{privacyGame}

\begin{thebibliography}{10}

\bibitem{blum2007learning}
Blum, A., Mansour, Y.:
\newblock {Learning, regret minimization, and equilibria}.
\newblock Algorithmic Game Theory (2007)  79--102

\bibitem{shapley1953stochastic}
Shapley, L.:
\newblock {Stochastic games}.
\newblock Proceedings of the National Academy of Sciences of the United States
  of America \textbf{39}(10) (1953)  1095

\bibitem{blocki2011regret}
Blocki, J., Christin, N., Datta, A., Sinha, A.:
\newblock Regret minimizing audits: A learning-theoretic basis for privacy
  protection.
\newblock In: Computer Security Foundations Symposium, 2011. CSF'11. 24th IEEE,
  IEEE (2011)  312--327

\bibitem{BCDS12}
Blocki, J., Christin, N., Datta, A., Sinha, A.:
\newblock Audit mechanisms for provable risk management and accountable data
  governance.
\newblock In: GameSec. (2012)

\bibitem{von2011market}
Von~StackelberG, H.:
\newblock Market structure and equilibrium.
\newblock Springer (2011)

\bibitem{AuditGames}
Blocki, J., Christin, N., Datta, A., Procaccia, A.D., Sinha, A.:
\newblock Audit games.
\newblock In: IJCAI. (2013)

\bibitem{tambe}
Tambe, M.:
\newblock Security and Game Theory: Algorithms, Deployed Systems, Lessons
  Learned.
\newblock Cambridge University Press (2011)

\bibitem{mertens1981stochastic}
Mertens, J., Neyman, A.:
\newblock {Stochastic games}.
\newblock International Journal of Game Theory \textbf{10}(2) (1981)  53--66

\bibitem{papadimitriou1999complexity}
Papadimitriou, C., Tsitsiklis, J.:
\newblock The complexity of optimal queueing network control (1999)

\bibitem{golovin2010AdaptiveSubmodularity}
Golovin, D., Krause, A.:
\newblock Adaptive submodularity: A new approach to active learning and
  stochastic optimization.
\newblock CoRR \textbf{abs/1003.3967} (2010)

\bibitem{even2005experts}
Even-Dar, E., Kakade, S., Mansour, Y.:
\newblock {Experts in a Markov decision process}.
\newblock In: Advances in neural information processing systems 17: proceedings
  of the 2004 conference, The MIT Press (2005)  401

\bibitem{mannor2003empirical}
Mannor, S., Shimkin, N.:
\newblock The empirical bayes envelope and regret minimization in competitive
  markov decision processes.
\newblock Mathematics of Operations Research (2003)  327--345

\bibitem{evenlearning}
Even-Dar, E., Mannor, S., Mansour, Y.:
\newblock Learning with global cost in stochastic environments.
\newblock In: COLT: Proceedings of the Workshop on Computational Learning
  Theory. (2010)

\bibitem{takimoto2003path}
Takimoto, E., Warmuth, M.:
\newblock Path kernels and multiplicative updates.
\newblock The Journal of Machine Learning Research \textbf{4} (2003)  773--818

\bibitem{awerbuch2008online}
Awerbuch, B., Kleinberg, R.:
\newblock Online linear optimization and adaptive routing.
\newblock Journal of Computer and System Sciences \textbf{74}(1) (2008)
  97--114

\bibitem{Farias2006}
Farias, D.P.D., Megiddo, N.:
\newblock Combining expert advice in reactive environments.
\newblock J. ACM \textbf{53} (September 2006)  762--799

\bibitem{fundenberggame}
Fudenberg, D., Tirole, J.:
\newblock Game theory.
\newblock MIT Press (1991)

\bibitem{blum2005external}
Blum, A., Mansour, Y.:
\newblock {From external to internal regret}.
\newblock Learning Theory (2005)  621--636

\bibitem{LongRun}
Celentani, M., Fudenberg, D., Levine, D., Pesendorfer, W.:
\newblock Maintaining a reputation against a patient opponent.
\newblock {Econometrica} \textbf{64} (1996)  691--704

\bibitem{impagliazzo2001complexity}
Impagliazzo, R., Paturi, R.:
\newblock On the complexity of k-sat.
\newblock Journal of Computer and System Sciences \textbf{62}(2) (2001)
  367--375

\bibitem{hastad2001some}
H{a}stad, J.:
\newblock {Some optimal inapproximability results}.
\newblock Journal of the ACM (JACM) \textbf{48}(4) (2001)  798--859

\bibitem{good1946normal}
Good, I.J.:
\newblock {Normal recurring decimals}.
\newblock Journal of the London Mathematical Society \textbf{1}(3) (1946)  167

\bibitem{auer1995gambling}
Auer, P., Cesa-Bianchi, N., Freund, Y., Schapire, R.:
\newblock {Gambling in a rigged casino: The adversarial multi-armed bandit
  problem}.
\newblock In: FOCS, Published by the IEEE Computer Society (1995)  322

\bibitem{littlestone1989weighted}
Littlestone, N., Warmuth, M.:
\newblock {The weighted majority algorithm}.
\newblock In: Proceedings of FOCS. (1989)  256--261

\bibitem{surveyregretstochastic08}
Yu, J.Y., Mannor, S.:
\newblock {(Survey) Online Learning in Stochastic Games and Markov Decision
  Processes} {http://www.cim.mcgill.ca/~jiayuan/survey08.pdf}.

\end{thebibliography}

\appendix
\section{Hardness Reduction: Proof of Claims} \label{apdx:Hardness}

This section contains the proofs of the lemmas and theorems from section \ref{sec:hardness}.

\newcommand{\claimHardnessFixedStrategyFullVersion}{
\begin{proofof}{Claim \ref{claim:hardnessFixedStrategy}}
Let $x_1*,...,x_{n-1}*$ be the assignment that satisfies at least $(1-\beta)$ fraction of the clauses and let $s_0,...,s_{n-1}$ be the De Bruijn sequence played by the adversary.  $x_n$ is an additional variable that is not in any of the clauses. Then the on round $t$ we have

\[ \sigma^t = \left( \langle  s_{i-1 \mod n},...,s_{i-m \mod n} \rangle, \hat{O}^{t-1} \right)    \ , \]

where $i = t \mod n$ so both these states are associated with the variable $x_i$.  For $0 \leq i < n$ we set

  \[ f\left( \langle  s_{i-1 \mod n},...,s_{i-m \mod n} \rangle, 0 \right)  = x_i* \ .\]

To avoid taking a penalty we set
\[f\left( \langle  s_{i-1 \mod n},...,s_{i-m \mod n} \rangle, 1 \right)   = 2 \ , \]
for $0 < i < n$.  For $i=0$ we set 
\[f\left( \langle  s_{i-1 \mod n},...,s_{i-m \mod n} \rangle, 1 \right)   = 0 \ , \]
to produce the outcome $\hat{O}^{t} = 0$ (recall that the adversary will play $a^t = (s_0,3)$ whenever $t \equiv 0 \mod n$ so we can avoid the penalty).  The fixed strategy $f$ will receive reward $1$ in stage $j$ if and only if $x_1*,...,x_{n-1}*$ satisfies the clause $C_j$ chosen in stage $j$. 

\begin{eqnarray}
 E_R \left[\AveragePay\left(f,\Adversary_R,G,n \right) \right] &\geq& \frac{(1-\beta)}{n} 
\end{eqnarray}

\end{proofof}}

\begin{reminderclaim}{\ref{claim:hardnessFixedStrategy}}
\claimHardnessFixedStrategy
\end{reminderclaim}

\claimHardnessFixedStrategyFullVersion

\newcommand{\proofClaimHardnessDefenderPerformance}{
\begin{proofof}{Claim \ref{claim:hardnessDefenderPerfomance}}
By Claim \ref{claim:hardnessFixedStrategy} there is a fixed strategy with

\[E_R \left[\AveragePay\left(\Defender,\Adversary_R,G,T \right) \right] \geq  \frac{(1-\beta)}{n} \ . \]

Set $\epsilon = \beta/n$, and apply definition \ref{def:RegMinAlgApproximate} to get
\[\AveragePay\left(f,\Adversary_R,G,T \right) -\AveragePay\left(\Defender,\Adversary_R,G,T \right) \leq \left(\frac{1}{8n} - \frac{3\beta}{n} \right) + \beta/n \ ,\]
for any random string $R$ (adversary coin flips).  This means that

\[E_R\left[\AveragePay\left(f,\Adversary_R,G,T \right)\right] -E_R\left[\AveragePay\left(\Defender,\Adversary_R,G,T \right)\right] \leq  \left(\frac{1}{8n} - \frac{3\beta}{n} \right) + \frac{\beta}{n} \ .\]

Rearranging terms
\begin{eqnarray*}
E_R\left[\AveragePay\left(\Defender,\Adversary_R,G,T \right)\right] &\geq& \frac{(1-\beta)}{n} - \frac{1}{8n} + \frac{2\beta}{n} \\ 
&=& \frac{7}{8n} + \frac{\beta}{n}
\end{eqnarray*}

\end{proofof}}
\begin{reminderclaim}{\ref{claim:hardnessDefenderPerfomance}}
\claimHardnessDefenderPerformance
\end{reminderclaim}

\proofClaimHardnessDefenderPerformance

\newcommand{\proofClaimRecoverAssignment}{
\begin{proofof}{Claim \ref{claim:recoverAssignment}}
Let $p(\circ)$ be given such that $T\left(\Defender \right) \leq p(n)$ and set
 \[\alpha = n \times E_R \left[\AveragePay\left(\Defender,\Adversary_R,G,T \right) \right] \ .\]

We present $S$  ( Algorithm \ref{alg:SimulationAdversary}) - an algorithm to recover the variable assignment.  $S$ runs in time \[T(S) = O\left(p(n)^2\right) \ . \]  

\begin{algorithm}[h!]
\caption{Assignment Recovery}\label{alg:SimulationAdversary}

\begin{itemize}
\item {\bf Input:} $\Defender$
  \item {\bf Input:} MAX3SAT instance $\phi$, with variables \[x_1,\ldots, x_{n-1} \ , \] and clauses \[C_1, \ldots , C_\ell \ , \]
  \item {\bf  De Bruin sequence: } $s_0,...,s_{n-1}$
  \item {\bf Initialize:} Set $t \leftarrow  0$, $H \leftarrow \emptyset$, $T \leftarrow p(n)$, $\alpha* \leftarrow 0$
  \item {\bf Round $t$: } Set $i \leftarrow t \mod n$ 
  
\begin{enumerate}
\item {\bf Check 1: } If $t \geq T$ then return.
\item {\bf Check 2: } If our current assignment $x_1,...,x_{n-1}$ satisfies $y$ fraction of the clauses where $y > \alpha*$ then set
 \[x_i* \leftarrow x_i \ ,\]
and
 \[ \alpha \leftarrow y \ . \]
  
\item  {\bf Select Clause:}  If $i = 0$ then select a new clause $C$ uniformly at random from $C_1,...,C_\ell$, and set $H' = \emptyset$. 
\item  {\bf Select Adversary Move: } 

\[ 
a^i  \leftarrow \begin{cases}
(s_i,3) & \text{if $i = 0$}; \\
(s_i,1) & \text{if $x_i \in C$};\\
(s_i,0) & \text{if $\bar{x}_i \in C$}; \\
(s_i,2) & \text{otherwise}.
\end{cases}
\]

\item {\bf Select Defender Move: }

\[d^i \leftarrow \Defender\left(H^{t-i}; H'\right) \ ,\]

\item {\bf Update:} Let $O^i$ be the outcome and set  

\[H \leftarrow H + \left(s_i, \hat{O}^i\right) \ , \]
\[H' \leftarrow H' + \left(s_i, 0\right) , \]
\[t \leftarrow t + 1 \ , \]
\[x_i \leftarrow d^i \ , \]

\end{enumerate}

\end {itemize}
\end{algorithm}

During the simulation we present $\Defender$ with (potentially) false history in each stage, where the defender always thinks he hasn't satisfied the clause $C$. Let $\mathbf{Y}_j$ be the expected fraction of clauses satisfied in stage $j$ of the simulation. We define the random variable $\mathbf{X}_j$ to be the reward $\Defender$ earns in stage $j$ in the actual game. Observe that the game is structured so that two rewards during the same stage must be separated by a penalty. When the defender receives a reward the outcome $\hat{O}^{t-1}$ is produced. If the defender wishes to avoid an offsetting penalty then he must keep producing the outcome $\hat{O}^{t-1}$ by playing $d^t = 2$, preventing him from receiving an award for the rest of the stage.  The maximum payout a defender strategy $\Defender$ can receive during any stage is $1$ so $\mathbf{X}_j \in \{0,1\}$. Because of imperfect information the defender cannot learn any information about the clause the adversary has selected. We have  

\[E[\mathbf{X}_j] = \Pr[\mathbf{X}_j=1] = E[\mathbf{Y}_j ] \ .\]

In particular

\[\alpha =  \frac{n}{T} \sum_{j=1}^{T/n} E[\mathbf{Y}_j ] \ , \]
so there exists a round $j$ such that $E[\mathbf{Y}_j ] \geq \alpha$.  Let $\mathbf{Y}$ denote the number of clauses satisfied by $S$, then \[\mathbf{Y} = \max_j \mathbf{Y}_j \ ,\] so we have
\[E[\mathbf{Y}] \geq \alpha \ . \]
\end{proofof}}

\begin{reminderclaim}{\ref{claim:recoverAssignment}}
\claimRecoverAssignment
\end{reminderclaim}

\proofClaimRecoverAssignment

\newcommand{\proofTheoremHardness}{

\begin{proofof}{Theorem \ref{thm:hardness}}
The key point is that if an algorithm $S$ runs in time $O\left(p(n)\right)$ on instances of size $n^\beta$ for some polynomial $p(n)$ then on instances of size $n$ $S$ runs in time $O\left(p\left(n^{1/\beta} \right) \right)$ which is still polynomial time. Unless {\sf NP} $=$ {\sf RP} $\forall \epsilon, \beta > 0$ and every algorithm $S$ running in time poly(n), there exists an integer $n$ and a MAX3SAT formula $\phi$ with $n^\beta$ variables such that 
\begin{enumerate}
\item There is an assignment satisfying at least $(1-\epsilon)$ of the clauses in $\phi$.
\item The expected fraction of clauses in $\phi$ satisfied by $S$ is  $\leq \frac{7}{8} + \epsilon$. 
\end{enumerate}
If we reduce from a MAX3SAT instance with $n^\beta$ variables we can construct a game with $O(n)$ states ($n^{1-\beta} $ copies of each variable). One Hamiltonian cycle would now corresponds to $n^{1-\beta} $ phases of the game. This means that the expected average reward of the optimal fixed strategy is at least
\[\max_{f \in \FixedDefender} E_R\left[ \AveragePay\left(f, \Adversary_R, G,T \right)\right] \geq \frac{n^{1-\beta} \left(1-\epsilon\right)}{n}  \ ,\]
while the expected average reward of an efficient defender strategy $\Defender$ is at most
\[   E_R\left[ \AveragePay\left(\Defender, \Adversary_R, G,T \right)\right] \leq \frac{n^{1-\beta} \left(\frac{7}{8}+\epsilon\right)}{n} \ .\]
Therefore, the expected average regret is at least
\[\AverageRegret{0}\left(\Defender, \Adversary_R, G,T,\FixedDefender \right) \geq \left(\frac{1}{8}-2\epsilon \right)n^{-\beta} \ . \]
\end{proofof}

}

\newcommand{\proofOfTheoremHardnessExpTime}{

\begin{proofof}{Theorem \ref{thm:hardnessExpTimeHyp}} (sketch)
Assume that the randomized exponential time hypothesis holds. Then because it is NP-hard to approximate MAX3SAT within any factor better than $\frac{7}{8}$ \cite{hastad2001some} no randomized algorithm which satisfies $\geq \frac{7}{8} + \epsilon$ of the clauses in a MAX3SAT instance in expectation can run in time \[ 2^{o(n)} \ . \]

Now we argue that it is sufficient to reduce from a MAX3SAT instance with $n' = \log^2 n$ variables (instead of $n^\beta$ variables). One Hamiltonian cycle now corresponds to \[\frac{n}{\log^{2} n} \ , \]
phases of the game. Our bounded-memory game $G$ has $n$ states then any efficient $\gamma$-approximate regret minimization algorithm $S$ must run in time $O\left(n^k\right)$ for some constant $k$. 
If the randomized exponential time hypothesis holds then the expected average reward of an efficient defender strategy $\Defender$ is at most
\[   E_R\left[ \AveragePay\left(\Defender, \Adversary_R, G,T \right)\right] \leq \frac{\frac{n}{\log^2 n} \left(\frac{7}{8}+\epsilon\right)}{n} \ ,\]
since \[ n^c = 2^{k \sqrt{\log^2 n}} = 2^{k \sqrt{n'}} = 2^{o(n')} \ . \]

However, if the MAX3SAT formula was satisfiable then the expected average reward of the optimal fixed strategy is at least
\[\max_{f \in \FixedDefender} E_R\left[ \AveragePay\left(f, \Adversary_R, G,T \right)\right] \geq \frac{\frac{n}{\log^2 n} \left(1-\epsilon\right)}{n} = \frac{1-\epsilon}{\log^2 n}  \ .\]

Therefore, the expected average regret is at least
\[\AverageRegret{0}\left(\Defender, \Adversary_R, G,T,\FixedDefender \right) \geq \frac{\left(\frac{1}{8}-2\epsilon \right)}{\log^2 n} \ . \]

Assume for contradiction that $\gamma < \frac{1}{8\log^2 n}$ then $S$ can be adapted to satisfy $\geq \frac{7}{8} + \epsilon$ of the clauses in MAX3SAT with running time
\[ n^c = 2^{k \sqrt{\log^2 n}} = 2^{k \sqrt{n'}} = 2^{o(n')} \ . \] This contradicts the randomized exponential time hypothesis.
 
\end{proofof}

}

Before we prove Theorem \ref{thm:hardness} we will first prove an easier Lemma using these claims. The proof of Lemma \ref{lemma:hardness} can be easily adapted to prove Theorems \ref{thm:hardness} and \ref{thm:hardnessExpTimeHyp}. 

\begin{lemma} \label{lemma:hardness}
\lemmaHardness
\end{lemma}

\proofofLemmaHardness

The proof of Theorem \ref{thm:hardness} is very similar to the proof of Lemma \ref{lemma:hardness}.

\begin{reminder}{\ref{thm:hardness}}
\thmHardness
\end{reminder}

\proofTheoremHardness

While the proof of Theorem \ref{thm:hardnessExpTimeHyp} makes use of the randomized exponential time hypothesis the argument is similar to the proof of Theorem \ref{thm:hardness}.

\begin{reminder}{\ref{thm:hardnessExpTimeHyp}}
\thmHardnessExpTimeHyp
\end{reminder}
\proofOfTheoremHardnessExpTime

Remark \ref{remark:HardnessExtension} shows how our hardness reduction can be adapted to prove that there is no efficient $k$-adaptive regret minimization algorithm in the perfect information setting $k \geq 1$.
\begin{remark} \label{remark:HardnessExtension} 
In bounded-memory games of perfect information we can replace the oblivious adversary  $\Adversary_R$ in figure \ref{alg:ObliviousAdversary} with a $1$-adaptive adversary and essentially the same reduction will still work. We only need to make a few small modifications. The states of the game will be modified to store the defenders last action. The adversary again plays a Hamiltonian cycle through the states in each phase. Now the first two states we visit correspond to the variable $x_1$, the next two visited states will correspond to $x_2$, etc. If the defender plays actions 1 and 1 (resp. 0 and 0) while visiting the variable $x_1$ then this corresponds to assigning $x_1$ to true (resp. false). If the defender plays 1 and 0 (or 0 and 1) which corresponds to no assignment then the adversary strategy will ensure that he cannot receive a reward. 

The $1$-adaptive adversary will always play $\vec{a}^t[2] = 2$ on even rounds ($t = 0 \mod{2}$) and on odd rounds the adversary will adaptively select $\vec{a}^t[2] = d^{t-1}$ if the defender's last action satisfied the chosen clause $C$, otherwise  $\vec{a}^t[2] = 2$. The defender receives a reward only if (1) he plays a consistent assignment during both rounds (2) the assignment satisfies the chosen clause $C$ and (3) he has not already received a reward during this phase. Now Claim \ref{claim:recoverAssignment} still holds because a defender will always observe the adversary action $\vec{a}^t[2] = 2$ until he satisfied the clause $C$.
\end{remark}

\subsection{Transition Example}
By playing a De Bruijn sequence $S = s_1...s_n$ the adversary can guarantee that we repeatedly take a Hamiltonian cycle over states. For example, considering $8$ states and starting from $x_0$, the sequence $10111000$ corresponds to the Hamiltonian cycle $x_0,x_1,x_2,x_5,x_3,x_7,x_6,x_4$ 
\begin{figure}

\begin{center}
\includegraphics[width=0.6\textwidth]{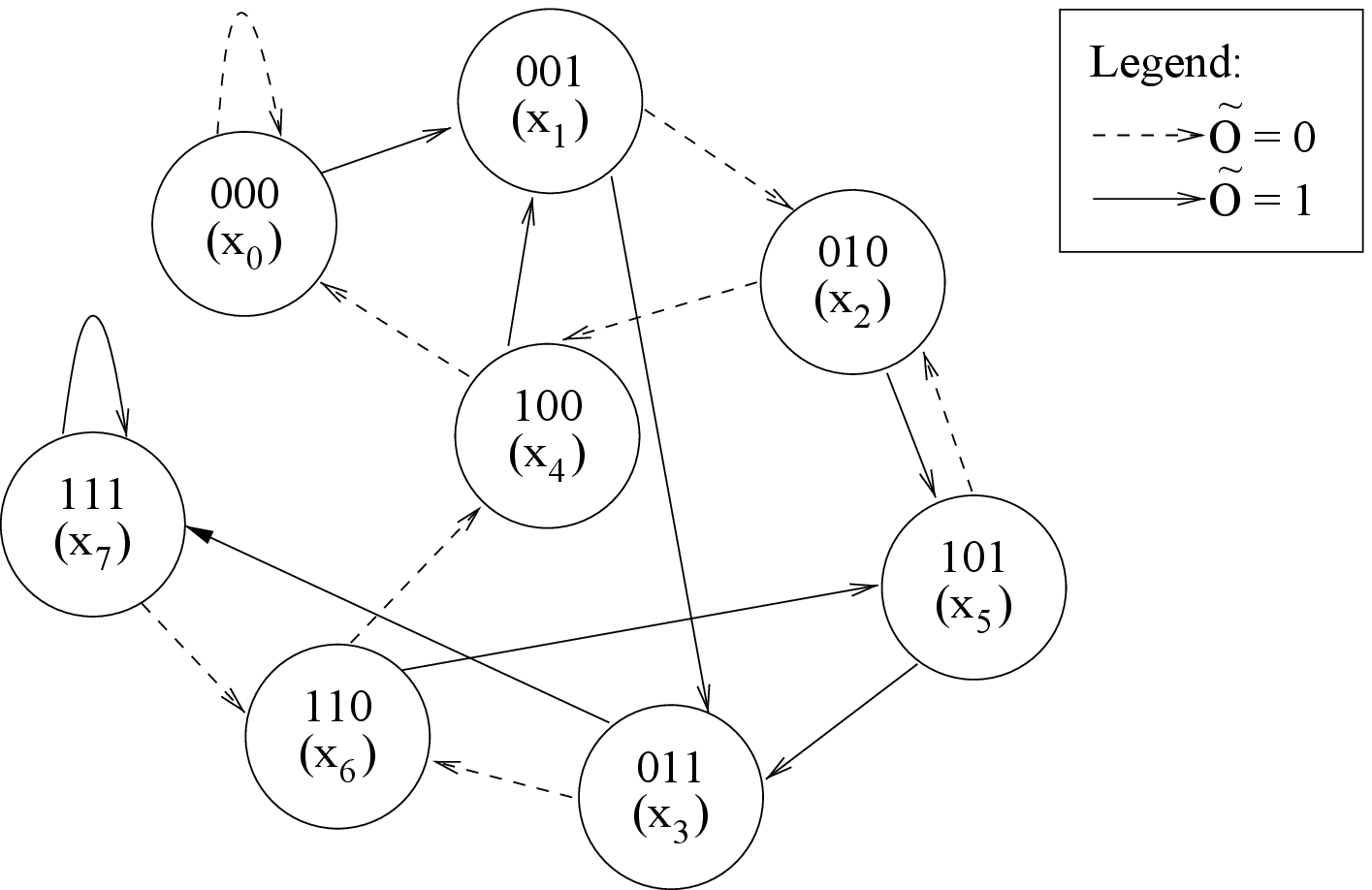}
\end{center}
\caption{De Bruijn example}

\end{figure}

\section{Regret Minimization Algorithms} \label{apdx:RegretMinAlg}

\newcommand{\proofTheoremBoundedMemoryRegretMinMain}{

\begin{proofof}{Theorem \ref{thm:memoryDTIRRegretMin}}
(Sketch)
The proof of theorem uses standard regret bound for regret minimization algorithms in games of perfect information \cite{auer1995gambling}. After playing $T$ rounds ($T/K$ rounds of $\RepeatedGame{G}{K}$) we have 
\begin{eqnarray*}
\AveragePay\left(\Defender, \KAdaptiveAdversary{k},\RepeatedGame{G}{K},T/K\right) - \AveragePay\left(f, \KAdaptiveAdversary{k},\RepeatedGame{G}{K},T/K\right)  &  \geq &  -4\sqrt{\frac{K N \log N}{T/K}} \ ,
\end{eqnarray*}
for all fixed strategies $f \in\FixedDefender$. Here, $N$ is the number of experts
\[ N = \left|\FixedDefender \right| = \left|\DefenderActions\right|^{\left|\Sigma\right|} \ , \]
and $K$ also denotes the maximum payout in any round of $\RepeatedGame{G}{K}$.Because $K$ was chosen such that $K \equiv 0 \mod{k}$ the adversary $\KAdaptiveAdversary{k}$ is always in phase with $\RepeatedGame{G}{K}$ and we can apply Claim \ref{claim:LossDifferences} to get Theorem
\ref{thm:memoryDTIRRegretMin}.
\end{proofof}}

\subsection{Regret Minimization Algorithm with Imperfect Information} \label{subsec:reduction}
We present $\BndMemRegMin$ (Bounded Memory Weighted Majority), an algorithm that minimizes $k$-adaptive regret for bounded-memory games. This result is significant because there is no $k$-adaptive regret minimization algorithm for the general class of stochastic games(see Theorem \ref{thm:stochasticImpossible} in Appendix \ref{apx:StochasticGames}). A consequence of Theorem \ref{thm:hardness} is that when the expert set includes all fixed strategies $\FixedDefender$ we cannot hope for an efficient algorithm unless {\sf NP} $=$ {\sf RP}.  Indeed, our algorithm would not be efficient in this case because it would have to explicitly maintains weights for exponentially many fixed strategies $\left|\FixedDefender\right| =  \left|\DefenderActions\right|^{n}$.

The key idea behind our $k$-adaptive regret minimization algorithm  $\BndMemRegMin$ is     
to reduce the original bounded-memory game to a repeated game $\RepeatedGame{G}{K}$ of      
imperfect information ($K \equiv 0 \mod{k}$). $\BndMemRegMin$ uses the $\mathsf{Exp3}$ regret minimization algorithm of \cite{auer1995gambling} for repeated games of imperfect information. In particular, $\BndMemRegMin$ uses the strategies selected by $\mathsf{Exp3}$ in each round of $\RepeatedGame{G}{K}$ to play the next $K$ rounds of $G$. $\BndMemRegMin$ feeds $\mathsf{Exp3}$ the hypothetical losses from $\RepeatedGame{G}{K}$ to update the weights of each expert.

\begin{reminder}{\ref{thm:memoryDTIRRegretMin}}
\ImpInfoRegMinAlgThm
\end{reminder}
\proofTheoremBoundedMemoryRegretMinMain

In particular, $\BndMemRegMin$ is a $k$-adaptive regret minimization algorithm for the class of bounded-memory games in the sense of Definition \ref{def:RegMinAlgApproximate} because $\AverageRegret{k} \rightarrow 0$ as $T \rightarrow \infty$. 

\begin{remark}
$\BndMemRegMin$ is inefficient when number of experts $f \in S$ is exponential in $n$, the number of states in $G$. For example, if $S=\FixedDefender$ then $\left| \FixedDefender\right|=\left|\DefenderActions \right|^n$.  For small values of $n$  (example: for repeated games $n=1$) it will still be tractable to run $\BndMemRegMin$ with $S = \FixedDefender$. 
\end{remark}

\subsection{Proofs of Claims and Theorems}
\label{apx:algproof}
This section contains the proof of claims and theorems from section \ref{sec:algorithm}. 

\newcommand{\proofClaimLossDifferences}{

\begin{proofof}{Claim \ref{claim:LossDifferences}}
(Sketch)
Once the defender selects $f$ and the adversary selects strategy $g \in K-ADAPT_A$, the actions of the adversary and the defender are fixed for the next $K$ rounds of $G$. Let $d^1,...,d^K$ (resp. $a^1,...,a^K$) denote the actions taken by the defender (resp. adversary). Once $R_1,...,R_K$ (the random coins used by the outcome function) are fixed then the outcomes $O^1,...,O^K$ are also fixed. Let $\sigma^1,...,\sigma^K$ states encountered in the actual game and let $\sigma_*^1, ..., \sigma_*^K$ be the states that we would have encountered if we had started at $\sigma_0$ as in $\RepeatedGame{G}{K}$. In a bounded-memory property game the state encodes the last $m$ outcomes, but the outcomes do not depend on the starting state so we have
\[ \sigma^j = \sigma_*^j  \ , \]
for all $j \geq m$. This means that for $j \geq m$
\[\Pay\left(\sigma^j, d^j, a^j \right) = \Pay\left(\sigma_*^j, d^j, a^j \right) \ . \]

 Consequently,

\begin{eqnarray*}
\left| \Pay\left(f,g,\sigma, G \right) - \Pay\left(f,g, \sigma_0, G\right) \right| &=&  \left| \sum_{t=1}^k \Pay\left(d_t, a_t, \sigma^i \right) - \sum_{t=1}^k \Pay\left(d_t, a_t, \sigma_*^i  \right) \right| \\
&=& \left| \sum_{t=1}^{m-1} \Pay\left(d_t, a_t, \sigma^i  \right) - \Pay\left(d_t, a_t, \sigma_*^i  \right)  \right| \\
&\leq& m \ .
\end{eqnarray*}

\end{proofof}}

\newcommand{\proofofTheoremPerfectInformation}{

\begin{proofof}{Theorem \ref{thm:PerfectInformation}}
By Claims \ref{claim:PerfInfoClaim} and \ref{claim:Sample}  Algorithm $\RegMinPerfInfo$ perfectly simulates the weighted majority algorithm \cite{littlestone1989weighted}. Notice that there are $N^n$ experts in $\PerfectInfoExpertSet$ and we are playing $T/K$ rounds of $\RepeatedGame{G}{K}$. The maximum payment in round of $\RepeatedGame{G}{K}$ is $K = m/\gamma$. The regret bound immediately follows from Claim \ref{claim:LossDifferences}  (the $\gamma = m/K$ term) and the standard regret bound from \cite{littlestone1989weighted} after setting 
\[\beta = \min \{\frac{1}{2}, \sqrt{\frac{n \ln \left(N\right)}{T}} \}  \ . \]

The regret bound holds against all experts $E \in \PerfectInfoExpertSet$ so in particular the regret bound also holds against all fixed experts $f \in \FixedDefender$ since $\FixedDefender \subset \PerfectInfoExpertSet$. 

The running time of $\RegMinPerfInfo$ is proportional to the number of traces in $\ConsistentAll.$ There are only  $n^{O\left(1/\gamma \right)}$ total traces in $\ConsistentAll$ so for any constant $\gamma$ the running time is polynomial.
\end{proofof}}

\newcommand{\proofOfSampleClaim}{

\begin{proof}
Given a trace $p = p_0; O; d$ let $\mathbf{Chosen}\left(p_0;O\right)$ be the event that the strategy output by Algorithm $\SampleAlgorithm$ plays $d$ from given history $p_0; O$. 

\begin{eqnarray*}
\Pr \left[\mbox{Output~} E \right] &=& \prod_{p \in \ConsistentAll, O \in \OutcomeSet} \Pr\left[ \mathbf{Chosen}\left( p;O \right) = E\left(p;O \right) \right] \\
&=& \prod_{p \in \ConsistentAll, O \in \OutcomeSet, d = E\left(p,O \right)} \frac{\hat{w}_{p;O;d}}{\sum_{d' \in \DefenderActions} \hat{w}_{p;O;d'}} \\
&=& \prod_{p \in \ConsistentAll, O \in \OutcomeSet, d = E\left(p,O \right)}  \frac{ \sum_{E':(p;O;d) \in \Consistent{E'}} \prod_{p' \in \Consistent{E'} \wedge p;O;d \sqsubset p' } w_{p'} }{\sum_{d' \in \DefenderActions} \sum_{E':(p;O;d') \in \Consistent{E'}} \prod_{p' \in \Consistent{E'} \wedge p;O;d' \sqsubset p' } w_{p'}} \\
&=& \prod_{p \in \ConsistentAll, O \in \OutcomeSet, d = E\left(p,O \right)}  \frac{ \sum_{E':(p;O;d) \in \Consistent{E'}} \prod_{p' \in \Consistent{E'} \wedge p;O;d \sqsubset p' } w_{p'} }{\sum_{d' \in \DefenderActions} \sum_{E':(p;O;d') \in \Consistent{E'}} \prod_{p' \in \Consistent{E'} \wedge p;O;d' \sqsubset p' } w_{p'}} \times \frac{\prod_{p' \sqsubset p } w_{p'}}{\prod_{p' \sqsubset p } w_{p'}} \\
&=& \prod_{p \in \ConsistentAll, O \in \OutcomeSet, d = E\left(p,O \right)}  \frac{ \sum_{E':(p;O;d) \in \Consistent{E'}} \prod_{p' \in \Consistent{E'}  } w_{p'} }{\sum_{d' \in \DefenderActions} \sum_{E':(p;O;d') \in \Consistent{E'}} \prod_{p' \in \Consistent{E'}  } w_{p'}} \\
&=& \prod_{p \in \ConsistentAll, O \in \OutcomeSet, d = E\left(p,O \right)}  \frac{ \sum_{E':(p;O;d) \in \Consistent{E'}} W_{E'} }{\sum_{d' \in \DefenderActions} \sum_{E':(p;O;d') \in \Consistent{E'}} W_{E'}} \\
&=& \prod_{p \in \ConsistentAll, O \in \OutcomeSet, d = E\left(p,O \right)}  \frac{ \sum_{E':(p;O;d) \in \Consistent{E'}} W_{E'} }{\sum_{E' \in \PerfectInfoExpertSet} W_{E'}} \\
&=& \frac{W_E}{\sum_{E' \in \PerfectInfoExpertSet} W_{E'}} \ .
\end{eqnarray*}   
\end{proof}}

\newcommand{\proofofTheoremImperfectInformation}{

\begin{proofof}{Theorem \ref{thm:ImperfectInformationApproximateRegret}}(Sketch)
We group the rounds of $\RepeatedGame{G}{K}$ into phases of $\frac{n^{1/\gamma}}{\gamma}$ rounds. Each phase now corresponds to 
\[ K\frac{n^{1/\gamma}}{\gamma} = \frac{m n^{1/\gamma}}{\gamma^2} \ , \]
rounds of $\BoundedMemoryGamesImperfectInfo$. As before there are $N^n$ experts. 

Within a single phase let $\vec{a}^i$ $\left(i=1,...,n^{1/\gamma}/\gamma\right)$ denote the actions of the adversary during round $i$ of that phase. To update our implicit weight representation we would like to compute  \[\sum_i \ell\left(p,\vec{a}^i,\sigma  \right) \ , \] for each $p \in \ConsistentAll$. However, we do not know the adversary actions $\vec{a}^i$ in each phase. Instead of computing  
 \[\sum_i \ell\left(p,\vec{a}^i,\sigma  \right) \ , \] 
we will estimate this quantity. For each \[\vec{d} \in \DefenderActions^\frac{m}{\gamma} \ , \] we will play the defender actions $\vec{d}$ in a randomly chosen round of the phase. Let $\vec{O}$ and $\vec{\ell} = \left(\ell_1,...,\ell_{m/\gamma}\right)$ denote the observed outcomes and payoffs in this round and let $p^j$ be the path corresponding to the first $j$ defender actions from $\vec{d}$ and outcomes from $\vec{O}$. For each path $p^j$ we set
\[\ell'\left(p^j,\sigma\right) = \frac{n^{1/\gamma}}{\gamma} \ell_j \ . \]
If the path $p$ never occured during a sampling round of the phase then we set
\[ \ell'\left(p^j,\sigma\right) = 0 \ .\]
For each path $p \in \ConsistentAll$ we have 
\begin{eqnarray*}
E\left[ \ell'\left(p,\sigma\right)\right] &=& \frac{n^{1/\gamma}}{\gamma}  E\left[\ell_i\right] \\ 
&=& \frac{n^{1/\gamma}}{\gamma} \sum_i \frac{\gamma}{n^{1/\gamma}} \ell\left(p,\vec{a}^i,\sigma  \right) \\
&=&    \sum_i  \ell\left(p,\vec{a}^i,\sigma  \right) \\
\end{eqnarray*}
where the expectation is taken over the random selection of sampling rounds. Now we can use the estimated losses $\ell'$ to maintain our implicit weight representation.

The following factors explain why the final regret bound is slightly worse than the bound in the perfect information setting (Theorem \ref{thm:PerfectInformation}): 

\begin{enumerate}
\item  We spend at most  
\[\left|\DefenderActions^\frac{m}{\gamma} \right| \leq n^{1/\gamma} \ ,\] rounds of each phase sampling. There are $ \frac{n^{1/\gamma}}{\gamma}$ rounds in a phase so the average sampling loss per round is at most
\[\frac{n^{1/\gamma}}{\left(\frac{n^{1/\gamma}}{\gamma}\right)} = \gamma \ . \]
This is in addition to modeling loss $\left(\gamma\right)$ from claim \ref{claim:LossDifferences}. In the perfect information setting there is no sampling loss just the modeling loss.
\item We are only now only updating weights after each phase. If $T$ is the number of rounds of the bounded-memory game $G$ that we play then we only update weights  $T'$ times where
\[ T' = \frac{T\gamma^2}{mn^{1/\gamma}} \ . \]
In the perfect information setting we had $T' = \frac{T\gamma}{m}$. 
\item The maximum loss in each phase is now the length of a phase
\[\frac{m}{\gamma} \left(\frac{n^{1/\gamma}}{\gamma} \right) \ , \]
instead of the length of a round $m/\gamma$.
\end{enumerate}

\end{proofof}}

\newcommand{\proofofClaimPerfectInfo}{

\begin{proofof}{Claim \ref{claim:PerfInfoClaim}}
First notice that we can write
\[ \sum_{j=1}^{T/K} \Pay\left(E, \vec{a}^j, \RepeatedGame{G}{K} \right) = \sum_{p \in \Consistent{E}} \sum_{j=1}^{T/K} \ell\left( p, \vec{a}^j, \sigma^{jK} \right) \ , \]
since the overall payoff of an expert $E$ can be expressed as a sum of the individual immediate payoffs after each action.

\begin{eqnarray*}
\prod_{p \in \Consistent{ E}} \beta^{\sum_{j=1}^{T/K} \ell\left(p,\vec{a}^j,\sigma^{jK}\right)} &=& \beta^{\sum_{p \in \Consistent{ E}} \sum_{j=1}^{T/K} \ell\left(p,\vec{a}^j,\sigma^{jK}\right)} \\
&=& \beta^{\sum_{t=1}^{T/K} \Pay\left(E, \vec{a}^t, \RepeatedGame{G}{K} \right)} \ .
\end{eqnarray*}

\end{proofof}}

Claim \ref{claim:LossDifferences} bounds the difference between the hypothetical losses from $\RepeatedGame{G}{K}$ and actual losses in $G$ using the bounded-memory property. 

\begin{reminderclaim}{\ref{claim:LossDifferences}}
\claimLossDifferences
\end{reminderclaim}

\proofClaimLossDifferences

The standard weighted majority algorithm maintains the invariant that 
$W_E = \beta^{\sum_{j=1}^{T/K} \Pay\left(E, \vec{a}^t, \RepeatedGame{G}{K}  \right)} $. Claim \ref{claim:PerfInfoClaim} says that $\RegMinPerfInfo$ also maintains this invariant. 

\begin{reminderclaim}{\ref{claim:PerfInfoClaim}}
\[
\prod_{p \in \Consistent{ E}} \beta^{\sum_{j=1}^{T/K} \ell\left(p,\vec{a}^j,\sigma^{jK}\right)} = \beta^{\sum_{j=1}^{T/K} \Pay\left(E, \vec{a}^j, \RepeatedGame{G}{K} \right)} \ . \]
\end{reminderclaim}

\proofofClaimPerfectInfo

\newcommand{\sampleClaim}{
Claim \ref{claim:Sample} says that $\SampleAlgorithm$ samples from the right distribution. 

\begin{reminderclaim}{ \ref{claim:Sample}}
For each expert $E \in \PerfectInfoExpertSet$ Algorithm $\SampleAlgorithm$ outputs $E$ with probability 
\[ \Pr\left[E \right] \propto W_E \ . \]
\end{reminderclaim}}

\sampleClaim

\proofOfSampleClaim

\begin{reminder}{\ref{thm:PerfectInformation}}
\RegMinPerfInfoThm
\end{reminder}

\proofofTheoremPerfectInformation 

\begin{reminder}{\ref{thm:ImperfectInformationApproximateRegret}}
\ImpInfoRegMinAlgThm
\end{reminder}

\proofofTheoremImperfectInformation

\begin{remark}
Because repeated games are a subset of bounded-memory games, 
$\RegMinPerfInfo$ (resp. $\RegMinImperfInfo$) could also be used to minimize oblivious regret in a repeated game of perfect information (resp. imperfect information)  using $\KAdaptiveDefenderSet$ as experts. In this case there is no modeling loss from claim \ref{claim:LossDifferences} so the guarantee is that we perform as well as the best K-adaptive defender strategy in hindsight. As long as $K = O\left(\log n \right)$ the running time of our algorithms will be time polynomial in $n$. 
\end{remark}

\cut{\begin{reminder}{\ref{thm:memoryDTIRRegretMin}}
\ImpInfoRegMinAlgThm
\end{reminder}}

\cut{
\begin{reminderclaim}{\ref{claim:PerfInfoClaim}}
\claimPerfInfoClaim
\end{reminderclaim}}

\cut{
\begin{reminder}{\ref{thm:PerfectInformation}}
\RegMinPerfInfoThm
\end{reminder}}

\section{Impossibility of Regret Minimization in Stochastic Games} \label{apx:StochasticGames}
\paragraph{Stochastic Games} 
Stochastic games are a generalization of repeated games, in which the   
payoffs depend on the state of play. Formally, a two-player stochastic  
game between an attacker~$A$ and a defender~$D$ is given by $(\DefenderActions, \AdversaryActions, \Sigma, \Pay, \tau)$, where $\AdversaryActions$ and $\DefenderActions$ are the actions spaces for players~$A$ and $D$,           
respectively, $\Sigma$ is the state space, $\Pay:  
\Sigma \times \DefenderActions \times  \AdversaryActions \rightarrow [0,1]$ is the payoff function and $\tau: \Sigma \times \DefenderActions \times  \AdversaryActions \times \{0,1\}^* \rightarrow \Sigma$ is the randomized transition    
function linking the different states. 

Thus, the payoff during round $t$ depends on the current state          
(denoted $\sigma^t$) in addition to the actions of the defender ($d^t$) and the adversary ($a^t$).  This added flexibility enables us to develop realistic game models    
for interactions where the rewards depend on game history. The          
hospital-employee interaction we introduced earlier is one example      
of such an interaction: an employee committing a given violation for    
the first time is unlikely to meet the same punishment as an employee   
committing the same violation for the tenth time.                       

A {\em fixed strategy} for the defender in a stochastic game is a function $f:\Sigma \rightarrow \DefenderActions$ mapping each state to a fixed action. $\FixedDefender$ denotes the set of all fixed strategies.

\label{sec:RegretMinStochasticGames}
In this section we demonstrate that there is no regret minimization algorithm  for the general class of stochastic games. More specifically for every notion of regret $k$ (oblivious ($k = 0$), $k$-adaptive, fully adaptive ($k = \infty)$) there is no $k$-adaptive minimization algorithm for the class of stochastic games.  It suffices to consider `oblivious regret' against an oblivious adversary (see remark \ref{remark:RegretMinAlgorithm}). The example in Theorem \ref{thm:stochasticImpossible} is fundamentally similar to example IV.1 of \cite{surveyregretstochastic08}.  

\begin{theorem}
\label{thm:stochasticImpossible}
There is a stochastic game $G$ such that for any defender strategies $\Defender$ there exists an oblivious adversary $\Adversary$ such that 
\[ \lim_{T \rightarrow \infty} \AverageRegret{k} \left(\Defender,\Adversary,G,T\right) > 0 \ .\]
\end{theorem}

\begin{proof}

In particular, consider the stochastic game $G$ illustrated in Figure \ref{fig:counterexample}.  The figure shows a game with two players D and A with action sets $\DefenderActions = \{d_1,d_2\}$ and  $ \AdversaryActions = \{a_1,a_2\}$ respectively. The reward function for the defender depends only on his own action as well as the current state $\sigma$.  Observe that $\sigma_2$ is a sink state which the game can never leave.  If the game reaches this state then the defender will be continuously rewarded in every round for the rest of the game.  However, the only way to reach $\sigma_2$ is if the defender and the adversary play $\left(d_1,a_1\right)$ simultaneously in some round $t$. If the defender fails to play $d_1$ then he might permanently miss his opportunity to reach $\sigma_2$.  This suggests that the defender must always play $d_1$. However, if the adversary never plays $a_1$ then it is best to use the fixed strategy always play $d_2$. 
\end{proof}

Notice that for any $\Adversary \in \KAdaptiveAdversarySet{0}$ and any defender strategy $\Defender$ we have 
\[\AveragePay\left(\Defender, \Adversary ,G, T \right) = \AveragePay\left(\Defender, \Adversary_0 ,G, T \right) \ \ , \]
because $\Adversary = \Adversary_0$. Hence, $\AverageRegret{0} = \AverageRegret{k}$ whenever the adversary is oblivious.

\begin{remark} \label{remark:RegretMinAlgorithm}
\smallskip 
\begin{enumerate}
\item If $\Defender$ can minimize k-adaptive regret against any k-adaptive adversary then $\Defender$ can minimize k-adaptive regret against any oblivious adversary ($k=0$) because
 \[ \KAdaptiveAdversarySet{0} \subset \KAdaptiveAdversarySet{k} \ .\]
\item If $\Defender$ can minimize k-adaptive regret against any k-adaptive adversary then $\Defender$ can minimize k-adaptive regret against any oblivious adversary because $\AverageRegret{0} = \AverageRegret{k}$ whenever the adversary is oblivious. 
\item If $\Defender$ is a k-regret minimization algorithm a class of games $\GameClass$ and  $\GameClass '$ is a subclass of  $\GameClass$ then $\Defender$ is also a k-regret minimization algorithm for the class of games $\GameClass '$.
\end{enumerate}
\end{remark}

\begin{figure}[ht] 
\begin{center}
\includegraphics[scale=0.75]{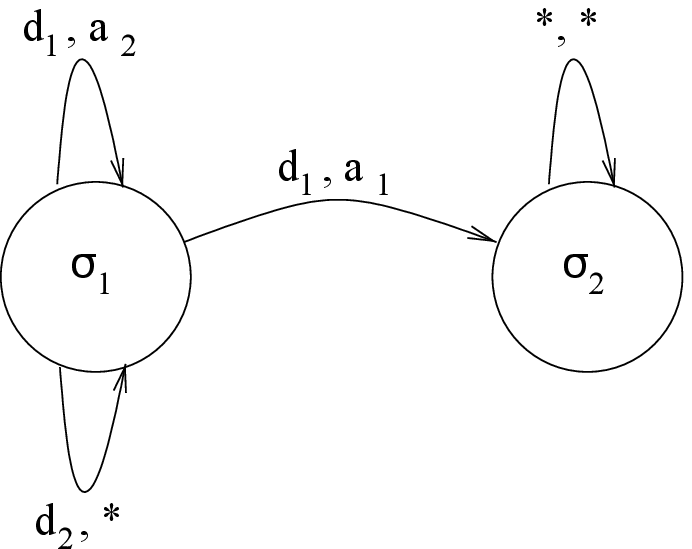}

\smallskip
\begin{tabular}{ll}
$P(d_1,\sigma_1) =-1$ & $P(d_1,\sigma_2) = 1$\\
$P(d_2,\sigma_1) = 0$ & $P(d_2, \sigma_2) = 1$
\end{tabular} 
\caption{A counterexample to prove Theorem \ref{thm:stochasticImpossible}}
\label{fig:counterexample}
\end{center}
\end{figure}

This example also illustrates why it is impossible to minimize fully adaptive regret against a non-forgetful adversary. In particular a non-forgetful adversary could use the states from \ref{fig:counterexample} to decide whether or not to cooperate. Note that even if the adversary can only see the last $m$ outcomes (sliding window) the adversary could play to remind himself of events arbitrarily long ago. For example, an adversary who wanted to remember whether or not the defender played action $d$ during round $1$ might play a special reminder action every $m$ rounds when the latest reminder is about to go out of memory.  


\end{document}